\documentclass[a4paper,fleqn]{cas-sc}

\usepackage[authoryear]{natbib}

\def\tsc#1{\csdef{#1}{\textsc{\lowercase{#1}}\xspace}}
\tsc{WGM}
\tsc{QE}

\usepackage{amssymb}
\usepackage{amsmath}
\usepackage{amsthm}
\usepackage{mathtools}
\usepackage{siunitx}
\usepackage{enumitem}
\usepackage{bm}
\usepackage{epsfig}
\usepackage{tikz}
\usetikzlibrary{arrows.meta,decorations.pathreplacing,calc,shadows,backgrounds}
\usepackage{caption, subcaption}
\usepackage{float}
\usepackage{hhline}
\usepackage{algorithm}
\usepackage{algpseudocode}
\usepackage{placeins}
\usepackage{booktabs}
\allowdisplaybreaks

\newcommand{\vx}{\bm{x}}
\newcommand{\vq}{\bm{q}}
\newcommand{\vf}{\bm{f}}
\newcommand{\vg}{\bm{g}}
\newcommand{\vu}{\bm{u}}
\newcommand{\vy}{\bm{y}}
\newcommand{\vh}{\bm{h}}
\newcommand{\vz}{\bm{z}}
\newcommand{\R}{\mathbb{R}}
\newcommand{\mPhi}{\bm{\Phi}}
\newcommand{\mPsi}{\bm{\Psi}}

\DeclareMathOperator{\range}{Range}

\newcommand{\mA}{\bm{A}}
\newcommand{\mG}{\bm{G}}
\newcommand{\mH}{\bm{H}}
\newcommand{\mV}{\bm{V}}
\newcommand{\mB}{\bm{B}}
\newcommand{\mI}{\bm{I}}

\newcommand{\mW}{\bm{W}}

\newcommand{\mC}{\bm{C}}
\newcommand{\mK}{\bm{K}}
\newcommand{\mQ}{\bm{Q}}
\newcommand{\mX}{\bm{X}}
\newcommand{\mY}{\bm{Y}}
\newcommand{\mR}{\bm{R}}
\newcommand{\mS}{\bm{S}}

\newcommand{\mZ}{\bm{Z}}
\newcommand{\mU}{\bm{U}}

\newcommand{\vv}{\bm{v}}

\newcommand{\mL}{\bm{L}}
\newcommand{\mSigma}{\bm{\Sigma}}

\newcommand{\mPi}{\bm{\Pi}}
\newcommand{\mM}{\bm{M}}
\newcommand{\mN}{\bm{N}}
\newcommand{\mT}{\bm{T}}
\newcommand{\vphi}{\bm{\phi}}
\newcommand{\vpsi}{\bm{\psi}}

\newcommand{\va}{\bm{a}}
\newcommand{\vb}{\bm{b}}

\DeclarePairedDelimiterX{\inp}[2]{\langle}{\rangle}{#1, #2}

\DeclareGraphicsExtensions{.pdf,.png}

\makeatletter
\def\thmheadbrackets#1#2#3{%
  \thmname{#1}\thmnumber{\@ifnotempty{#1}{ }\@upn{#2}}%
  \thmnote{ {\the\thm@notefont[#3]}}}
\makeatother

\newtheoremstyle{brakets}
  {}
  {}
  {\itshape}
  {}
  {\bfseries}
  {.}
  { }
  {\thmheadbrackets{#1}{#2}{#3}}

\theoremstyle{brakets}

\newtheorem{theorem}{Theorem}
\newtheorem{definition}{Definition}

\begin{document}
\let\WriteBookmarks\relax
\def\floatpagepagefraction{1}
\def\textpagefraction{.001}

\shorttitle{}    

\shortauthors{}  

\title [mode = title]{Non-Intrusive Balancing for Data-Driven Model Reduction with Oblique Projections}



%

\author[1]{Alberto Padovan}



\ead{alberto.padovan@njit.edu}

\credit{Conceptualization, Methodology, Software, Funding Acquisition}

\affiliation[1]{organization={Department of Mechanical and Industrial Engineering, New Jersey Institute of Technology},
            addressline={200 Central Ave.}, 
            city={Newark},
         citysep={}, 
            postcode={07114}, 
            state={NJ},
            country={USA}}




\begin{abstract}
Non-intrusive reduced-order modeling is essential for systems simulated with black- or gray-box solvers, yet computing accurate data-driven models remains challenging when the dynamics exhibit large-amplitude transient energy growth. This is because state-of-the-art non-intrusive methods typically rely on orthogonal projection onto high-energy linear or nonlinear manifolds identified by minimizing static reconstruction losses, which truncates low-energy, yet dynamically important, states. Intrusive methods such as Balanced Truncation and its extensions address this issue by using adjoint information to define dynamics-aware oblique projections. Here, we show that balancing-based oblique projections for nonlinear systems can instead be computed non-intrusively from forward simulation data. Specifically, we develop a framework to compute the nonlinear empirical reachability Gramian and balance it against its observability counterpart, projected onto a suitable restriction subspace. This procedure provably recovers the oblique projection of Balanced Truncation when the dynamics are linear and the range of either Gramian is contained in the restriction subspace. Moreover, in the linear case and in the absence of adjoint information, we show that the minimax-optimal restriction subspace spans the dominant eigenspaces of the reachability Gramian, which informs an efficient balancing algorithm with cost comparable to a standard SVD. Once the oblique projection is computed, we use Operator Inference to learn the latent-space dynamics on its range. The method is demonstrated on several examples of increasing complexity, including the three-dimensional incompressible flow past a wing with $8\times 10^7$ degrees of freedom. In all examples, the Balanced Operator Inference models significantly outperform their POD-based counterparts, reducing forecasting error by up to two orders of magnitude.
\end{abstract}



\begin{keywords}
Non-intrusive model reduction \sep Oblique projections \sep Balanced Truncation \sep Operator Inference\sep
\end{keywords}

\maketitle

\section{Introduction}
\label{sec:intro}

Model reduction via non-intrusive algorithms is essential to enable complex engineering tasks such as control, optimization, forecasting, and uncertainty quantification for high-dimensional physical systems that are simulated via black- or gray-box solvers with proprietary source codes.
While non-intrusive reduced-complexity methods have seen a rapid growth over the past several years, computing dynamically-accurate models remains an open challenge, especially for systems that exhibit large deviations from underlying attractors when excited by external disturbances.
It is well known that the solution to this problem lies in the identification of appropriate latent-space coordinates via dynamics-aware encoders capable of capturing the sensitivity mechanisms responsible for the transient dynamics of the original system.
However, such mechanisms are typically exposed through adjoint calculations that are inherently intrusive and thus limited to open-source solvers with linearization/automatic-differentiation capabilities.
In non-intrusive settings, where adjoint information is unavailable, the question of how to identify dynamically-appropriate coordinates remains largely open and is the focus of this work.

We begin by reviewing the state of the art of both intrusive and non-intrusive model reduction methods. 
Most frameworks rely on three ingredients: a map from the high-dimensional state space to a low-dimensional latent space (i.e., an encoder), a map from the latent space to the original state space (i.e., a decoder), and intrinsic dynamics that evolve the low-dimensional state in the latent space.
Classical Petrov-Galerkin models share this architecture, where the encoder $\mPsi^\top$ and the decoder $\mPhi(\mPsi^\top\mPhi)^{-1}$ are linear and define an oblique projection $\mathbb{P} = \mPhi(\mPsi^\top\mPhi)^{-1}\mPsi^\top$, while the latent-space dynamics are obtained by projecting the full-order dynamics onto the range of $\mPhi$ along directions orthogonal to the range of $\mPsi$ \citep{rowley_model_2017}.
When $\mPhi = \mPsi$, the projection is orthogonal and the model is known as a Galerkin model.
In the simplest of cases, $\mPhi$ is given by the leading POD modes of a representative data set \citep{lumley1967structure,berkoozProperOrthogonalDecomposition1993,holmes2012turbulence}, and POD-Galerkin models have been used extensively in fluid mechanics \citep{rowley_model_2004,baroneStableGalerkinReduced2009}.
These models, however, are known to perform poorly in systems that exhibit large transient energy growth in response to external disturbances.
This loss of performance is due to the fact that the energy-based truncation of the state space performed via POD may accidentally discard low-energy directions that are nonetheless important for forecasting, and that are responsible for the transient growth behavior exhibited by the original full-order model.
In linear systems, or in nonlinear systems that evolve in the proximity of a steady state, this issue can be addressed using Balanced Truncation \citep{Moore1981principal} or its approximations \citep{willcox2002balanced,rowley2005model}, where dynamics-aware oblique projections are computed by simultaneously diagonalizing the reachability and observability Gramians of the underlying linear system.
Balancing has also been extended to periodically time-varying systems \citep{sandberg2004balanced,padovan2024jcp} and nonlinear systems \citep{SCHERPEN1993,kramer2023nonlinearbalancedtruncation2}, although nonlinear balancing remains limited to relatively low-dimensional settings due to the curse of dimensionality associated with the computation of the controllability and observability energy functions \citep{KRAMER2024117011}.
Other popular approaches compute oblique projections by minimizing the $\mathcal{H}_2$ or $\mathcal{H}_\infty$ norm of the error between the full-order and reduced-order transfer functions \citep{antoulas2020interpolatory}, by solving least-squares problems at each time step as in Least-Squares Petrov-Galerkin \citep{carlberg2011efficient,carlberg2017galerkin}, by balancing state and gradient covariances \citep{otto2023cobras,Zanardi2025,Zanardi2026}, or by optimizing Petrov-Galerkin models against high-fidelity data \citep{otto2022troop}.
All these methods are, however, intrusive, as they require access to the full-order dynamics and, crucially, to their linearization and its adjoint.
As previously mentioned, many high-fidelity commercial and legacy solvers across engineering disciplines do not have linearization and adjoint capabilities, thus limiting the applicability of intrusive methods.

On the non-intrusive front, a popular framework for learning continuous-time reduced-order models is Operator Inference \citep{peherstorfer_data-driven_2016,kramer_learning_2024}, which fits polynomial latent-space dynamics to data by solving a linear least-squares problem.
Operator Inference has been extended to systems with non-polynomial nonlinearities \citep{bennerOperatorInferenceNonintrusive2020}, systems with Lagrangian or Hamiltonian structure \citep{sharma2022hamiltonian,sharma2024lagrangian,gruber2025variationally}, systems exhibiting continuous symmetries \citep{shuai2025}, and models that are globally asymptotically stable by construction \citep{goyal_guaranteed_2025,gkimisisRepresentationEnergypreservingQuadratic2026}.
Despite these developments, Operator Inference almost invariably achieves dimensionality reduction via orthogonal projection onto the span of the leading POD modes, and it therefore inherits the shortcomings of POD-Galerkin models when applied to systems with large-amplitude transient growth.
Notable exceptions include the work of \cite{gruber2025variationally} in the context of Hamiltonian dynamics, where the Operator Inference problem is closed via Petrov-Galerkin projection in order to guarantee variational consistency of the resulting model, and the more recent work of \cite{rettberg2026petrov} in the context of linear port-Hamiltonian systems, where the authors designed a structure-aware Petrov-Galerkin formulation by endowing the ambient space with a non-Euclidean Hamiltonian-informed inner product.
In order to address the problem of learning dynamics-aware oblique projections for general nonlinear systems, the recently introduced NiTROM formulation \citep{padovan2024nitrom,errico2026gasnitrom} poses a trajectory-based optimization problem where the projectors and the (possibly globally-asymptotically-stable) latent-space dynamics are learned simultaneously.
This formulation has been shown to achieve significant improvements over POD-based formulations on several fluid flows, but the optimization problem is strongly non-convex and thus requires careful initialization.
Along the lines of non-intrusive, dynamics-aware projections, we also find the recent work of \cite{bettini2025datadriven,bettini2026general}, where reduced-order models are obtained by projecting the state onto spectral submanifolds along their stable invariant fibers (or linear approximations thereof).
Similarly to the NiTROM formulation, this approach relies on solving a non-convex optimization problem to simultaneously identify an appropriate encoder and intrinsic dynamics of polynomial form.
Additionally, in \cite{bettini2026general}, the learning problem is formulated using delay embeddings of a scalar-valued observable of the state, and it is unclear how the method might scale in high dimensions if one chooses to observe the whole state.

In this paper, we propose a non-intrusive and \emph{optimization-free} procedure to compute balancing-based oblique projections for nonlinear systems from forward simulation data alone.
The procedure builds on the empirical Gramians of \cite{Lall1999}, and it rests on the observation that while the computation of the empirical observability Gramian as defined in \cite{Lall1999} is intractable in high dimensions, its restriction to a suitably chosen subspace can be computed inexpensively from the outputs of a few forward simulations.
This leads to the definition of a \emph{projected} empirical observability Gramian which, once balanced against the empirical reachability Gramian, yields the desired dynamics-aware encoder-decoder pairs.
We show that this procedure provably recovers the oblique projection defined via classical Balanced Truncation \citep{Moore1981principal} when the dynamics are linear and when the range of either the reachability or the observability Gramian is contained in the aforementioned restriction subspace.
When neither range is contained in this subspace, we quantify the deviation from Balanced Truncation in terms of the Hankel singular values and of the principal angles between the trial and test subspaces, measured in the geometry induced by the reachability Gramian.
Furthermore, in the linear case and in the absence of adjoint information, we show that under a notion of minimax optimality the restriction subspace should span the dominant eigenspaces of the reachability Gramian, and this result informs the development of an efficient, data-driven algorithm to perform non-intrusive balancing of the empirical Gramians.
We emphasize that the cost of this algorithm is comparable to that of a standard singular value decomposition (SVD) of the training data, both in terms of storage and CPU/GPU time.
Once the desired balanced encoder-decoder pairs are computed, we learn latent-space dynamics via Operator Inference \citep{peherstorfer_data-driven_2016} or its globally-asymptotically-stable variant \citep{goyal_guaranteed_2025}.
We remark, however, that any other available method for learning latent-space dynamics can be used, including Neural ODEs \citep{chen2018neuralode}, SINDy \citep{brunton2016sindy}, and several others.
We demonstrate the proposed method on several examples of increasing size and complexity. These include the complex Ginzburg-Landau equation in both stable and unstable configurations, the two-dimensional flow past an airfoil, and the three-dimensional flow past a finite wing with approximately $8\times 10^7$ degrees of freedom.
Throughout, the proposed method achieves significant improvements over POD-based Operator Inference.

The manuscript is organized as follows.
In Section~\ref{sec:math_formulation}, we introduce the data-driven balancing procedure, we analyze its connection with Balanced Truncation in the linear setting, we present the checkpointing algorithm, and we summarize the Operator Inference procedure used to learn the latent-space dynamics.
The numerical examples are presented in Section~\ref{sec:results}, and conclusions are drawn in Section~\ref{sec:conclusions}.

\section{Mathematical Formulation}
\label{sec:math_formulation}

We begin by considering a nonlinear full-order model (FOM) governed by
\begin{equation}
  \label{eq:nonlinear_fom}
  \begin{aligned}
    \frac{d}{dt}\vx(t) &= \vf(\vx(t),\vu(t)),\quad \vx(0) = \vx_0, \\
    \vy(t) &= \vh(\vx(t)),
  \end{aligned}
\end{equation}
where $\vx(t)\in\R^n$ is the state, $\vu(t)\in\R^m$ is the input, and  $\vy(t)\in\R^p$ is the output. 
Throughout this paper, we seek reduced-order models (ROMs) capable of accurately predicting the output $\vy(t)$ across a wide range of inputs $\vu(t)$ and initial conditions $\vx_0$.
We restrict our attention to models of the form
\begin{equation}
  \label{eq:nonlinear_rom}
  \begin{aligned}
    \frac{d}{dt}\vz(t) &= \vf_r(\vz(t),\vu(t)),\quad \vz(0) = \vz_0 \coloneqq \mPsi^\top\vx_0, \\
    \hat\vy(t) &= \vh\left(\mPhi \left(\mPsi^\top \mPhi\right)^{-1}\vz(t)\right),
  \end{aligned}
\end{equation}
where $\vz(t)\in\R^r$, with $r\ll n$, is the reduced state, and $\hat\vy(t)$ is the predicted output.
The model parameters, which we intend to learn from high-fidelity data sampled from \eqref{eq:nonlinear_fom}, are the nonlinear reduced-order dynamics $\vf_r$, the encoder $\mPsi^\top:\R^n\to\R^r$, and the decoder $\mPhi\left(\mPsi^\top \mPhi\right)^{-1}:\R^r\to V$, where $V = \range(\mPhi)$ is an~$r$-dimensional subspace of $\R^n$.
We emphasize that the encoder and decoder are linear and that the encoder is a left-inverse of the decoder, so that the decoder-encoder pair defines a rank-$r$ linear projection operator $\mathbb{P}\coloneqq\mPhi\left(\mPsi^\top \mPhi\right)^{-1}\mPsi^\top : \R^n \to V$.
We note in passing that with this construction we immediately recover a classical Petrov-Galerkin ROM if $\vf_r(\vz(t),\vu(t)) = \mPsi^\top\vf(\mPhi(\mPsi^\top \mPhi)^{-1}\vz(t),\vu(t))$.
Additionally, if $\range(\mPsi) = \range(\mPhi)$, then $\mathbb{P}$ is an orthogonal projector, and projecting the full-order dynamics \eqref{eq:nonlinear_fom} would lead to a classical Galerkin ROM.

In order to learn the model parameters, we propose a two-stage learning framework, whereby the encoder and decoder are learned via a novel data-driven balancing procedure, and the latent-space dynamics are learned via Operator Inference.
We remark that the focus of this paper is on the balancing procedure used to identify dynamics-aware encoder-decoder pairs, and the latent-space dynamics (or solution operator) can then be inferred using any available method in the literature.
This section is structured as follows.
In Subsection~\ref{subsec:learning_encdec}, we describe the balancing procedure. In particular, we start by specializing to linear systems in order to draw a parallel with the classical Balanced Truncation approach, and then we present the general nonlinear formulation that is used in practice.
In Subsection~\ref{subsec:learning_lsd}, we describe the Operator Inference procedure of \cite{peherstorfer_data-driven_2016} to learn the latent-space dynamics.

\subsection{Learning Linear Encoder/Decoder Pairs}
\label{subsec:learning_encdec}

Here, we describe a data-driven procedure to learn the encoder and decoder matrices $\mPsi$ and $\mPhi$ from high-fidelity data.
While the procedure is designed for general nonlinear dynamics (Section~\ref{subsubsec:data_driven_nonlinear_balancing}), it is instructive to first examine its specialization to linear dynamics (Section~\ref{subsubsec:data_driven_linear_balancing}), where the empirical Gramians reduce to the classical Gramians of Balanced Truncation \citep{Moore1981principal,willcox2002balanced,rowley2005model}, and where the procedure admits a transparent interpretation and rigorous error bounds.
Finally, we present a computationally-efficient algorithm to compute the desired trial and test bases (Section~\ref{subsubsec:computational_algorithm}).

\subsubsection{Non-Intrusive Balancing for Linear Dynamics}
\label{subsubsec:data_driven_linear_balancing}

We begin by considering linear full-order dynamics and output maps of the form $\vf(\vx(t),\vu(t)) = \mA\vx(t) + \mB\vu(t)$ and $\vh(\vx(t)) = \mC\vx(t)$, so that the full-order model in \eqref{eq:nonlinear_fom} is a linear time-invariant input-output system.
We also assume that the system is asymptotically stable, so that the infinite-time Gramians (which we will define shortly) are well-defined.
Now, let us suppose that we have collected $m$ linear impulse responses $\{\vx^{(j)}(t)\} = \{e^{\mA t}\mB_j\}$, where $\mB_j$ is the $j$th column of $\mB$.
The infinite-time reachability Gramian associated with the linear dynamics may then be defined as
\begin{equation}
  \label{eq:linear_WR}
  \mW_R = \sum_{j=1}^m \int_0^\infty \vx^{(j)}(t) \vx^{(j)}(t)^\top dt = \sum_{j=1}^m \int_0^\infty e^{\mA t}\mB_j\mB_j^\top e^{\mA^\top t} dt = \int_0^\infty e^{\mA t}\mB\mB^\top e^{\mA^\top t} dt.
\end{equation}
An analogous definition of the observability Gramian would normally require collecting impulse responses from the adjoint system, which is not available in a non-intrusive setting such as the one considered here.
We then ask whether any observability information can be extracted from forward trajectory data.
To do so, consider the data set $\{\vy^{(j)}(t)\} = \{\mC e^{\mA t}\vx_0^{(j)}\}$, with $j\in\{1,2,\ldots,l\}$, where $\vx_0^{(j)}$ are $l$ initial conditions that we can choose freely.
Choosing such initial conditions is a design choice that we will discuss in more detail in Section~\ref{subsubsec:computational_algorithm}.
Then, consider the $l\times l$ Gram matrix associated with the measured outputs $\vy(t)$,
\begin{equation}
  \mG_{i,j} = \int_0^\infty \vy^{(i)}(t)^\top \vy^{(j)}(t) dt = \vx_0^{(i)\top}\underbrace{\int_0^\infty e^{\mA^\top t}\mC^\top\mC e^{\mA t} dt}_{\mW_O}\, \vx_0^{(j)} = \vx_0^{(i)\top} \mW_O \vx_0^{(j)}\quad \Rightarrow\quad \mG = \mX_0^\top \mW_O\mX_0,
\end{equation}
where $\mX_0$ is the $n\times l$ matrix whose columns are the initial conditions $\vx_0^{(j)}$.
Interestingly, we see that $\mG$ is the compressed representation of the observability Gramian $\mW_O$ on the range of $\mX_0$.
Indeed, under the assumption that $\mX_0$ is full column rank, we may then define a projected observability Gramian 
\begin{equation}
  \label{eq:linear_WO}
  \widetilde{\mW}_O = \mX_0\left(\mX_0^\top \mX_0\right)^{-1}\mG\left(\mX_0^\top \mX_0\right)^{-1}\mX_0^\top = \mathbb{P}_{\mX_0} \mW_O \mathbb{P}_{\mX_0},\quad \mathbb{P}_{\mX_0} = \mX_0\left(\mX_0^\top \mX_0\right)^{-1}\mX_0^\top.
\end{equation}
At this point, the reachability and (projected) observability Gramians can be balanced using an approach analogous to the one described in \cite{rowley2005model}.
Specifically, let $\mW_R = \mX \mX^\top$ and $\widetilde{\mW}_O = \widetilde{\mY}\widetilde{\mY}^\top$ be factorizations of the reachability and projected observability Gramians, respectively, with $\widetilde{\mY} = \mX_0\left(\mX_0^\top \mX_0\right)^{-1} \mL$, where $\mL$ is the Cholesky factor of~$\mG$.
Then, a balancing transformation may be computed from the singular value decomposition of $\widetilde{\mY}^\top \mX = \widetilde{\mU} \widetilde{\mSigma} \widetilde{\mV}^\top$, yielding the desired matrices
\begin{equation}
  \label{eq:PhiPsi}
  \widetilde{\mPhi} = \mX \widetilde{\mV}_r \widetilde{\mSigma}_r^{-1/2},\quad \widetilde{\mPsi} = \widetilde{\mY} \widetilde{\mU}_r \widetilde{\mSigma}_r^{-1/2},
\end{equation}
each of dimension $n \times r$, where $\widetilde{\mV}_r$ denotes the first $r$ columns of $\widetilde{\mV}$.
The tildes are used to emphasize the fact that the factors come from balancing the reachability and \emph{projected} observability Gramians.
This distinction will be important in the discussion that follows, where we will compare the results of this procedure with those of classical Balanced Truncation, whose factors do not carry tildes.
The theorem below establishes conditions under which the two procedures are equivalent, and it provides bounds on the deviation from Balanced Truncation when such conditions are not met.

\begin{definition}
  For matrices $\mM$ and $\mN$ with full column rank, $\theta_{\max}(\mM,\mN)$ denotes the largest principal angle between $\range(\mM)$ and $\range(\mN)$, and
  \begin{equation}
    \sin\theta_{\max}(\mM,\mN) = \big\lVert\left(\mI - \mathbb{P}_{\mM}\right)\mN\left(\mN^\top\mN\right)^{-1/2}\big\rVert,\quad \mathbb{P}_{\mM}\coloneqq\mM\left(\mM^\top\mM\right)^{-1}\mM^\top.
  \end{equation}
  Moreover, given a symmetric positive-definite matrix $\mS$, we let $\theta^{\mS}_{\max}(\mM,\mN)\coloneqq\theta_{\max}(\mS^{1/2}\mM,\mS^{1/2}\mN)$ denote the largest principal angle between $\range(\mM)$ and $\range(\mN)$ measured in the inner product $\inp{\va}{\vb}_{\mS}\coloneqq\va^\top\mS\vb$.
\end{definition}

\begin{theorem}
  \label{thm:balanced_truncation_bounds}
  Let $\mW_R = \mX \mX^\top$, $\mW_O = \mY \mY^\top$, and $\widetilde{\mW}_O = \widetilde{\mY}\widetilde{\mY}^\top$ be factorizations of the reachability, observability, and projected observability Gramians, respectively, with $\widetilde{\mY} = \mPi\mY$ and $\mPi \coloneqq \mathbb{P}_{\mX_0}$, so that $\widetilde{\mW}_O = \mPi\mW_O\mPi$.
  Likewise, let $\mH = \mY^\top \mX$ and $\widetilde{\mH} = \widetilde{\mY}^\top \mX$ be the corresponding Hankel matrices, and let $\sigma_k(\mH)$ and $\sigma_k(\widetilde{\mH})$ denote their singular values.
  With $\delta^2_R \coloneqq \lVert \mPi_\perp \mW_R \mPi_\perp\rVert$ and $\delta^2_O \coloneqq \lVert \mPi_\perp \mW_O \mPi_\perp\rVert$, where $\mPi_\perp = \mI - \mPi$ and $\lVert \cdot\rVert$ denotes the spectral norm, the following inequality holds
  \begin{equation}
    \label{eq:weyl_bound}
    \big\lvert\sigma_k\left(\widetilde{\mH}\right) - \sigma_k\left(\mH\right)\big\rvert \leq \delta_O\delta_R,\quad \forall k \in \{1,2,\ldots,\text{rank}(\mH)\}.
  \end{equation}
  Moreover, let $\mH = \mU\mSigma\mV^\top$ and $\widetilde{\mH} = \widetilde{\mU}\widetilde{\mSigma}\widetilde{\mV}^\top$ be singular value decompositions, and define the Balanced Truncation trial and test matrices as
  \begin{equation}
    \mPhi = \mX\mV_r\mSigma_r^{-1/2},\quad \mPsi = \mY\mU_r\mSigma_r^{-1/2},
  \end{equation}
  where the subscript $r$ denotes truncation to rank $r$, and let $\widetilde{\mPhi}$ and $\widetilde{\mPsi}$ be defined as in \eqref{eq:PhiPsi}.
  Then, provided $\varepsilon\coloneqq \sigma_r(\widetilde{\mH}) - \sigma_{r+1}(\mH) > 0$, and assuming that~$\mW_R$ is positive definite, we have
  \begin{equation}
    \label{eq:sin_theta_bound}
    \sin\theta^{\mW_R^{-1}}_{\max}\left(\mPhi,\widetilde{\mPhi}\right) = \sin\theta_{\max}\left(\mV_r,\widetilde{\mV}_r\right) \leq \frac{\delta_O\delta_R}{\varepsilon},\quad \sin\theta^{\mW_R}_{\max}\left(\mPsi,\widetilde{\mPsi}\right) = \sin\theta_{\max}\left(\mV_r,\widetilde{\mV}_r\right) \leq \frac{\delta_O\delta_R}{\varepsilon}.
  \end{equation}
\end{theorem}
\begin{proof}
  The proof relies on Weyl's inequality and Wedin's $\sin\theta$ theorem, and it is provided in Appendix~\ref{sec:proof_of_thm_1}.
\end{proof}

In words, the first half of the theorem establishes that the singular values of the Hankel matrix $\widetilde{\mH}$ computed from the projected observability Gramian are close to those of the Hankel matrix $\mH$ computed from the full observability Gramian, with a bound that depends on how well the reachability and observability Gramians are captured by the projection onto the range of $\mX_0$.
The second half establishes that a similar quantity bounds the deviation between the trial and test subspaces identified by the classical and projected Balanced Truncation procedures.
These deviations are measured in the geometry induced by the Gramians themselves.
Specifically, the norm $\lVert\vx\rVert_{\mW_R^{-1}}^2 = \vx^\top\mW_R^{-1}\vx$ is the minimum input energy required to steer the system from the origin to the state $\vx$ \citep{hespanha2018linear}, and in this geometry the trial subspaces deviate exactly as much as the right singular subspaces of the Hankel matrices.
The test subspaces, which define the encoder $\vx\mapsto\widetilde{\mPsi}^\top\vx$ acting on states, are instead measured in the geometry induced by~$\mW_R$, and they deviate by exactly the same angle.
This geometry is meaningful because the norm $\lVert\cdot\rVert_{\mW_R}$ is dual to the minimum-energy norm $\lVert\cdot\rVert_{\mW_R^{-1}}$, so that for any row $\widetilde{\vpsi}^\top$ of the encoder we have, by definition,
\begin{equation}
  \lVert\widetilde{\vpsi}\rVert_{\mW_R} = \max_{\lVert\vx\rVert_{\mW_R^{-1}}\leq 1}\ \widetilde{\vpsi}^\top\vx.
\end{equation}
That is, the norm of $\widetilde{\vpsi}$ in this geometry is the maximum value that the corresponding reduced coordinate can take among all states $\vx$ that are reachable with unit input energy.
Consequently, two encoders that are close in this geometry produce nearly identical reduced coordinates on strongly reachable states, and their differences are confined to weakly reachable (or unreachable) states that therefore contribute little to the input-output map.

The bounds in the theorem are particularly useful not only because they provide a quantitative measure of the accuracy of the projected model, but also because they can be leveraged to inform the choice of initial conditions $\vx_0^{(j)}$ used to compute the projected observability Gramian.
In particular, it is revealing to observe that the bound $\varepsilon^{-1}\delta_O \delta_R$ vanishes when either the range of the reachability Gramian $\mW_R$ or the range of the observability Gramian $\mW_O$ is contained within that of $\mX_0$ (i.e., when $\delta_R = 0$ or $\delta_O = 0$, respectively).
While the range of the observability Gramian is generally not known in a non-intrusive setting, the range of the reachability Gramian is well approximated by the span of the impulse responses $\{\vx^{(j)}(t)\}$, which are already available from the data collection stage.
The theorem therefore suggests a sampling strategy of $\mX_0$ along impulse responses.
In order to formalize this intuition, we use Theorem~\ref{thm:minimax} below to show that, among all possible choices of $\mX_0$ with rank at most $\ell$, the one that minimizes the maximum error between the Hankel matrices $\mH$ and $\widetilde{\mH}$ is obtained by letting the range of $\mX_0$ coincide with the dominant $\ell$-dimensional eigenspace of $\mW_R$.
Since the dominant eigenspace of $\mW_R$ is (approximately) contained in the range of the impulse responses (this can be seen directly from equation \eqref{eq:linear_WR}), this result provides a theoretical motivation for the numerical algorithm described in Section~\ref{subsubsec:computational_algorithm}.

\begin{theorem}\label{thm:minimax}
Let $\mW_R = \mX\mX^\top$ be positive definite with eigenvalues $\lambda_1\geq\dots\geq\lambda_n$ and eigenvectors $\vu_1,\dots,\vu_n$, let $\widetilde{\mH}$ denote the Hankel matrix computed from the projected observability Gramian associated with $\mX_0$, and let $\mathcal{O}_\gamma := \{\mW_O\succeq 0 : \|\mW_O\|\leq\gamma\}$. Then
\begin{equation}
\min_{\mathrm{rank}(\mX_0)\leq \ell}\;\sup_{\mW_O\in\mathcal{O}_\gamma}\;\big\|\mH-\widetilde{\mH}\big\| = \gamma^{1/2}\lambda_{\ell+1}^{1/2},
\end{equation}
and the minimum is attained by any $\mX_0$ with $\range(\mX_0) = \mathrm{span}\{\vu_1,\dots,\vu_\ell\}$.
\end{theorem}

\begin{proof}
Let $\range(\mX_0) = \mathrm{span}\{\vu_1,\dots,\vu_\ell\}$ and $\mPi := \mathbb{P}_{\mX_0}$. Then $\mPi_\perp\mW_R\mPi_\perp = \sum_{k>\ell}\lambda_k\vu_k\vu_k^\top$, so that $\delta_R^2 = \lambda_{\ell+1}$. Since $\delta_O^2 = \|\mPi_\perp\mW_O\mPi_\perp\|\leq\|\mW_O\|\leq\gamma$, the first half of the proof of Theorem~\ref{thm:balanced_truncation_bounds} gives $\|\mH-\widetilde{\mH}\|\leq\delta_O\delta_R\leq\gamma^{1/2}\lambda_{\ell+1}^{1/2}$ for all $\mW_O\in\mathcal{O}_\gamma$.
This establishes the maximum error attained by letting the range of $\mX_0$ coincide with the dominant $\ell$-dimensional eigenspace of $\mW_R$.
We now show that no other choice of $\mX_0$ with $\mathrm{rank}(\mX_0)\leq \ell$ leads to a smaller maximum error.
Let $\mathrm{rank}(\mX_0)\leq \ell$ and $\mPi := \mathbb{P}_{\mX_0}$. Since $\dim\range(\mPi_\perp)\geq n-\ell$, there exists a unit vector $\vv\in\range(\mPi_\perp)\cap\mathrm{span}\{\vu_1,\dots,\vu_{\ell+1}\}$, so that $\mPi\vv = \mathbf{0}$ and $\vv^\top\mW_R\vv\geq\lambda_{\ell+1}$. Taking $\mW_O = \gamma\vv\vv^\top\in\mathcal{O}_\gamma$, with factor $\mY = \gamma^{1/2}\vv$, we have $\widetilde{\mY}\widetilde{\mY}^\top = \mPi\mW_O\mPi = 0$, so that $\widetilde{\mH} = 0$ and
\begin{equation*}
\big\|\mH-\widetilde{\mH}\big\| = \|\mH\| = \gamma^{1/2}\|\vv^\top\mX\| = \gamma^{1/2}\big(\vv^\top\mW_R\vv\big)^{1/2}\geq\gamma^{1/2}\lambda_{\ell+1}^{1/2}.
\end{equation*}
\end{proof}

\subsubsection{Non-Intrusive Balancing for Nonlinear Dynamics}
\label{subsubsec:data_driven_nonlinear_balancing}

We now present the procedure in its general form, which applies to nonlinear dynamics and follows the ``Empirical Gramians'' paradigm of \cite{Lall1999}.
For notational simplicity, we shall consider control-affine nonlinear dynamics of the form $\vf(\vx,\vu) = \vg(\vx) + \mB\vu$, but this assumption can be relaxed straightforwardly.
We do assume, however, that the origin is globally asymptotically stable for the unforced dynamics $d\vx/dt = \vg(\vx)$, so that infinite-time integrals are well-defined.
If the origin is unstable and the dynamics settle onto a different attractor, the ``nonlinear'' Gramians defined below can be defined over a finite-time horizon, and the balancing procedure and algorithmic details remain unchanged.
As in the previous section, consider collecting nonlinear impulse responses $\{\vx^{(j,k)}(t)\}$ by evolving \eqref{eq:nonlinear_fom} with initial conditions $\vx^{(j,k)}(0) = \alpha_k \mB_j$, where $\alpha_k$ are real scalars and $\mB_j$ is the $j$th column of $\mB$.
Then, we define the empirical reachability Gramian as
\begin{equation}
  \label{eq:empirical_WR}
  \mW_R =  \frac{1}{s}\sum_{k=1}^s \mW_R^{(k)},\quad \mW_R^{(k)} = \frac{1}{\alpha_k^2} \sum_{j=1}^m\int_0^\infty \vx^{(j,k)}(t) \vx^{(j,k)}(t)^\top dt,
\end{equation}
so that $\mW_R$ can be understood as an average of partial Gramians $\mW^{(k)}_R$ associated with impulse responses with initial amplitude scaling $\alpha_k$.
This definition ensures that $\mW_R$ in equation \eqref{eq:empirical_WR} reduces to the linear reachability Gramian in~\eqref{eq:linear_WR} when the dynamics are linear and asymptotically stable, since $\mW_R^{(k)} = \mW_R^{(1)}$ for all $k$.
Additionally, we note that this definition is very similar to the one in \cite{Lall1999}, with the main difference being that we restrict the initial conditions to be along the input directions~$\mB_j$, whereas \cite{Lall1999} allow for arbitrary linear combinations of the latter.
In the examples considered in this paper, we do not find this additional layer of generality necessary, and we therefore omit it for simplicity.

Following a similar procedure, we can define the (projected) empirical observability Gramian by collecting nonlinear output responses $\{\vy^{(j,k)}(t)\}$ from trajectories of \eqref{eq:nonlinear_fom} started from (freely chosen) initial conditions $\vx^{(j,k)}(0)$ of norm $\lVert \vx^{(j,k)}(0)\rVert = \beta_k$, where $j\in \{1,2,\ldots,l\}$ and $k\in \{1,2,\ldots, q\}$.
Then, the Gramian is defined as
\begin{equation}
  \label{eq:empirical_WO}
  \widetilde{\mW}_O = \frac{1}{q}\sum_{k=1}^q \mX^{(k)}_0 \left(\mX_0^{(k)\top}\mX_0^{(k)}\right)^{-1}\mG^{(k)}\left(\mX_0^{(k)\top}\mX_0^{(k)}\right)^{-1}\mX_0^{(k)\top}\quad\text{with}\quad \mG^{(k)}_{i,j} = \frac{1}{\beta_k^2} \int_0^\infty \vy^{(i,k)}(t)^\top \vy^{(j,k)}(t) dt.
\end{equation}
Here, $\mX_0^{(k)}$ is the $n\times l$ matrix whose $j$th column is the unit-norm initial condition $\vx^{(j,k)}(0)/\beta_k$.
Once again, it is straightforward to show that we recover the (projected) linear observability Gramian in \eqref{eq:linear_WO} when the dynamics and output map $\vh$ are linear.
Additionally, we remark that in the case of $\mX^{(k)}_0 = \mI\in\R^{n\times n}$, equation \eqref{eq:empirical_WO} recovers the empirical observability Gramian in \cite{Lall1999}.
We stress, however, that setting $\mX^{(k)}_0 = \mI$ is equivalent to running $n$ full-order simulations (where $n$ is the state-space dimension), which is infeasible in high-dimensional settings.
The use of a (carefully chosen) set of initial conditions $\mX_0^{(k)}\in\R^{n\times l}$ is therefore crucial to making the balancing procedure tractable in practice, and it is the key observation that makes the non-intrusive balancing procedure applicable to high-dimensional nonlinear systems.
Given the factorizations of these nonlinear empirical Gramians into $\mW_R = \mX\mX^\top$ and $\widetilde{\mW}_O = \widetilde{\mY}\widetilde{\mY}^\top$, we can then proceed to compute a balancing transformation as in the linear case, yielding the desired encoder-decoder pair in equation \eqref{eq:PhiPsi}.
A data-efficient algorithm for the computation of these factorizations is discussed in the next subsection.

\subsubsection{A Computationally-Efficient Implementation via Checkpointing}
\label{subsubsec:computational_algorithm}

We now develop a computationally-efficient algorithm to compute a suitable factorization of the nonlinear empirical Gramians~$\mW_R$ and $\widetilde{\mW}_O$.
We begin by observing that the empirical reachability Gramian $\mW_R$ can be approximated numerically by storing the impulse responses $\vx^{(j,k)}(t)$ at discrete time instances $t_i$ and performing numerical quadrature.
More specifically, $\mW_R$ in equation \eqref{eq:empirical_WR} can be factored numerically into $\mW_R = \mX \mX^\top$ following the approach of \cite{rowley2005model}, where
\begin{equation}
  \label{eq:X_factor}
  \mX = s^{-1/2}\begin{bmatrix}
    \alpha_1^{-1}\mX^{(1)} & \alpha_2^{-1}\mX^{(2)} & \cdots & \alpha_s^{-1}\mX^{(s)}
  \end{bmatrix}\in\R^{n\times s\, m\, N},
\end{equation}
and
\begin{equation}
  \label{eq:X_blocks}
  \mX^{(k)} = \begin{bmatrix}
    \mX^{(1,k)} & \mX^{(2,k)} & \cdots & \mX^{(m,k)}
  \end{bmatrix},\quad \mX^{(j,k)} = \begin{bmatrix}
    \sqrt{\gamma_1}\vx^{(j,k)}(t_1) & \sqrt{\gamma_2}\vx^{(j,k)}(t_2) & \cdots & \sqrt{\gamma_{N}}\vx^{(j,k)}(t_{N})
  \end{bmatrix}.
\end{equation}
Here, $0 = t_1 < t_2 < \ldots < t_N$ are the sampling instances and $\gamma_i$ are quadrature weights for the evaluation of the temporal integral in \eqref{eq:empirical_WR}, which is truncated at $t = t_N$.
In the case of a uniform sampling interval $\Delta t = t_i - t_{i-1} = t_2 - t_1$, we have $\gamma_i = \Delta t$.

Factoring the projected empirical observability Gramian $\widetilde{\mW}_O$ in \eqref{eq:empirical_WO} can be done in a similar fashion, but first we must define the initial condition matrices $\mX_0^{(k)}$.
The approach relies on the result of Theorem~\ref{thm:minimax}, which, in the linear case, suggests that the initial conditions should be chosen so that the range of $\mX_0^{(k)}$ coincides with the dominant eigenspace of the reachability Gramian $\mW_R^{(k)}$.
Since the dominant eigenspace of the reachability Gramian is well approximated by the span of the impulse responses $\{\vx^{(j,k)}(t)\}$ (see equation \eqref{eq:linear_WR}), we conclude that the initial conditions for the projected empirical observability Gramian should be sampled along the forward trajectories.
We therefore proceed as follows.
Consider the (uniformly sampled) trajectory $\{\vx^{(j,k)}(t_i)\}_{i=1}^N$ generated by the $j$th impulse response of amplitude $\alpha_k$, and let $\mX_0^{(k)}$ be defined as
\begin{equation}
  \label{eq:checkpoints}
  \begin{aligned}
    \mX_0^{(k)} &= \begin{bmatrix}
      \mX_0^{(1,k)} & \mX_0^{(2,k)} & \cdots & \mX_0^{(m,k)}
    \end{bmatrix}\in\R^{n\times m N_c},\,\,
    \mX_0^{(j,k)} = \begin{bmatrix}
    \frac{\vx^{(j,k)}(0)}{\lVert \vx^{(j,k)}(0) \rVert} & \frac{\vx^{(j,k)}(\Delta\tau)}{\lVert \vx^{(j,k)}(\Delta\tau) \rVert} & \cdots & \frac{\vx^{(j,k)}((N_c-1)\Delta \tau)}{\lVert \vx^{(j,k)}((N_c-1)\Delta \tau) \rVert}
  \end{bmatrix}\in\R^{n\times N_c},
  \end{aligned}
\end{equation}
where $\Delta\tau$ is a user-defined integer multiple of $\Delta t = t_{i} - t_{i-1}$, and $N_c$ is the user-defined number of \emph{checkpoints}, so that the $c$th checkpoint is located at $t = (c-1)\Delta\tau$ for $c\in\{1,2,\ldots,N_c\}$.
Since the checkpoints are extracted from the impulse responses used to compute $\mW_R$, the empirical observability Gramian \eqref{eq:empirical_WO} is evaluated with the same $q = s$ amplitude levels as the empirical reachability Gramian, with $\beta_k = \lvert\alpha_k\rvert$, and with $l = m N_c$ initial conditions per amplitude level (namely, $N_c$ checkpoints from each of the $m$ impulse responses).
In other words, this checkpointing procedure allows us to treat existing FOM snapshots as initial conditions for the computation of the empirical observability Gramian, without requiring additional FOM simulations.
The original trajectory $\{\vx^{(j,k)}(t_i)\}_{i=1}^N$ and its measured outputs $\{\vy^{(j,k)}(t_i)\}_{i=1}^N$ are then windowed and scaled into $N_c$ subtrajectories defined over a temporal window of length $\Delta T \coloneqq t_{N} - (N_c-1)\Delta\tau$,
\begin{equation}
  \label{eq:windowing}
  \vx^{(j,k,c)}(t) = \frac{\alpha_k}{\lVert \vx^{(j,k)}((c-1)\Delta\tau) \rVert}\vx^{(j,k)}(t + (c-1)\Delta\tau),\quad t \in [0, \Delta T],
\end{equation}
so that the $c$th subtrajectory starts from the $c$th column of $\mX_0^{(j,k)}$ scaled to have norm $\lvert\alpha_k\rvert$.
The checkpointing procedure is illustrated schematically in Figure~\ref{fig:checkpointing}.
\begin{figure}
  \centering
  \begin{tikzpicture}[
    >=Stealth,
    font=\small,
    declare function={
      f(\t) = exp(-\t/1.5) + 2.4*(\t/2.5)*exp(1 - \t/2.5);
      g(\t,\l) = exp(-\t/1.5) + (2.4 - 0.25*\l)*(\t/(2.5 + 0.45*\l))*exp(1 - \t/(2.5 + 0.45*\l));
    },
    ckpt/.style={circle, fill=#1, draw=white, line width=0.8pt, inner sep=2pt},
    axis/.style={-{Stealth[length=2mm]}, line width=0.6pt, black!65},
    panel/.style={rounded corners=6pt, fill=black!2.5, draw=black!12, line width=0.5pt},
    title/.style={anchor=south west, font=\small},
    box/.style={rounded corners=4pt, draw=#1!50, fill=#1!5, line width=0.6pt, inner sep=6pt, align=center,
                drop shadow={opacity=0.18, shadow xshift=0.7pt, shadow yshift=-0.7pt}},
    link/.style={-{Stealth[length=2.4mm]}, line width=0.9pt, black!45},
    note/.style={font=\footnotesize\itshape, black!55},
    mask/.style={fill=black!2.5, inner sep=1pt, rounded corners=1pt},
  ]
  \definecolor{c1}{RGB}{0,114,178}
  \definecolor{c2}{RGB}{213,94,0}
  \definecolor{c3}{RGB}{0,158,115}
  \definecolor{c4}{RGB}{204,121,167}
  \definecolor{ink}{RGB}{40,40,40}
  \def\sx{0.50}     
  \def\sy{0.85}     
  \def\dd{0.75}     
  \def\tilt{-15}    
  \def\aang{150}    
  \def\dtau{1.3}
  \def\tN{9}
  \pgfmathsetmacro{\DT}{\tN - 3*\dtau}
  \pgfmathsetmacro{\exx}{\sx*cos(-\tilt)}
  \pgfmathsetmacro{\exy}{\sx*sin(-\tilt)}
  \pgfmathsetmacro{\eyx}{\dd*cos(\aang)}
  \pgfmathsetmacro{\eyy}{\dd*sin(\aang)}

  \begin{scope}[local bounding box=A,
      x={(\exx cm,\exy cm)},
      y={(\eyx cm,\eyy cm)},
      z={(0cm,\sy cm)}]
    \fill[black!6] (0,0,0) -- (\tN,0,0) -- (\tN,2,0) -- (0,2,0) -- cycle;
    \draw[black!22, line width=0.4pt] (\tN,0,0) -- (\tN,2,0);
    \foreach \l/\op in {2/0.35, 1/0.6}{
      \fill[black!2.5] (0,\l,0) -- plot[domain=0:\tN, samples=80, smooth] (\x,\l,{g(\x,\l)}) -- (\tN,\l,0) -- cycle;
      \fill[black!7, opacity=\op] (0,\l,0) -- plot[domain=0:\tN, samples=80, smooth] (\x,\l,{g(\x,\l)}) -- (\tN,\l,0) -- cycle;
      \draw[black!30, line width=0.5pt] (0,\l,0) -- (\tN,\l,0);
      \draw[black!25, line width=0.4pt] (0,\l,0) -- (0,\l,{g(0,\l)});
      \draw[black!25, line width=0.4pt] (\tN,\l,0) -- (\tN,\l,{g(\tN,\l)});
      \draw[ink, line width=0.9pt, smooth, opacity=\op] plot[domain=0:\tN, samples=80] (\x,\l,{g(\x,\l)});
      \foreach \c/\col in {0/c1, 1/c2, 2/c3, 3/c4}{
        \pgfmathsetmacro{\ts}{\c*\dtau}
        \fill[\col, opacity=\op] (\ts,\l,{g(\ts,\l)}) circle (1.5pt);
      }
    }
    \fill[black!2.5] (0,0,0) -- plot[domain=0:\tN, samples=120, smooth] (\x,0,{f(\x)}) -- (\tN,0,0) -- cycle;
    \fill[black!7] (0,0,0) -- plot[domain=0:\tN, samples=120, smooth] (\x,0,{f(\x)}) -- (\tN,0,0) -- cycle;
    \draw[axis, black!50] (0,0,0) -- (0,2.9,0) node[above left=-1pt, black!65] {$\lvert\alpha_k\rvert$};
    \foreach \l/\n in {0/1, 1/2, 2/3}{
      \fill[black!55] (0,\l,0) circle (0.9pt);
      \node[anchor=north east, black!60, font=\footnotesize, inner sep=1.5pt] at (0,\l,0) {$\lvert\alpha_{\n}\rvert$};
    }
    \draw[axis] (0,0,0) -- (\tN+0.6,0,0) node[below=1pt] {$t$};
    \draw[axis] (0,0,0) -- (0,0,4.45) node[above] {$\lVert\vx^{(j)}\rVert$};
    \draw[black!30, line width=0.4pt] (\tN,0,0) -- (\tN,0,{f(\tN)});
    \draw[ink, line width=1.2pt, smooth] plot[domain=0:\tN, samples=120] (\x,0,{f(\x)});
    \foreach \c/\col in {0/c1, 1/c2, 2/c3, 3/c4}{
      \pgfmathsetmacro{\ts}{\c*\dtau}
      \pgfmathsetmacro{\te}{\ts + \DT}
      \pgfmathsetmacro{\zz}{-0.8 - 0.26*\c}
      \draw[\col, densely dotted, line width=0.6pt] (\ts,0,{f(\ts)}) -- (\ts,0,\zz);
      \draw[\col!30, line width=4.2pt, line cap=round] (\ts,0,\zz) -- (\te,0,\zz);
      \draw[\col, line width=0.8pt, line cap=round] (\ts,0,\zz) -- (\te,0,\zz);
      \fill[\col] (\ts,0,\zz) circle (1.5pt);
      \node[ckpt=\col] at (\ts,0,{f(\ts)}) {};
    }
    \foreach \c/\col/\lab in {0/c1/0, 1/c2/\Delta\tau, 2/c3/2\Delta\tau, 3/c4/3\Delta\tau}{
      \node[anchor=north west, \col, mask, xshift=1pt, yshift=-1pt] at (\c*\dtau,0,0) {$\lab$};
    }
    \draw[black!65, line width=0.6pt] (\tN,0,0) -- (\tN,0,-0.12);
    \node[anchor=north, mask, yshift=-1pt] at (\tN,0,-0.12) {$t_N$};
    \draw[decorate, decoration={brace, mirror, amplitude=4pt}, black!55]
      (3*\dtau,0,-1.75) -- (\tN,0,-1.75) node[midway, below=5pt, sloped, note] {windows of length $\Delta T$};
  \end{scope}

  \begin{scope}[local bounding box=B, shift={($(A.east)+(2.55,0.3)$)}]
    \draw[black!30, densely dotted, line width=0.6pt] (0,0.55) -- (\DT*0.7, 0.55);
    \draw[axis] (0,0) -- (\DT*0.7 + 0.35, 0) node[right] {$\tau$};
    \draw[axis] (0,0) -- (0, 2.05) node[above] {$\lVert\vx^{(j,k,c)}\rVert$};
    \foreach \c/\col in {3/c4, 2/c3, 1/c2, 0/c1}{
      \pgfmathsetmacro{\ts}{\c*\dtau}
      \pgfmathsetmacro{\fs}{f(\ts)}
      \draw[\col, line width=1.1pt, smooth] plot[domain=0:\DT, samples=100] (\x*0.7, {0.55*f(\x + \ts)/\fs});
    }
    \node[ckpt=ink] at (0, 0.55) {};
    \node[anchor=east] at (-0.08, 0.55) {$\lvert\alpha_k\rvert$};
    \node[below] at (0, 0) {$0$};
    \node[below] at (\DT*0.7, 0) {$\Delta T$};
    \foreach \c/\col [count=\i] in {1/c1, 2/c2, 3/c3, 4/c4}{
      \draw[\col, line width=1.1pt, line cap=round] (\DT*0.7 + 0.6, {2.0 - 0.3*\i}) -- ++(0.32,0) node[right, ink] {$c = \c$};
    }
    \node[note, align=left, anchor=north west] at (-0.1, -0.75) {each window is shifted to $\tau = 0$ and\\ rescaled to have initial norm $\lvert\alpha_k\rvert$};
  \end{scope}

  \coordinate (top) at ($(A.north)+(0,0.25)$);
  \coordinate (bot) at ($(A.south)+(0,-0.25)$);
  \node[title] (TA) at ($(A.west |- top)+(-0.1,0)$) {\textbf{(a)} Stored impulse responses $\vx^{(j,k)}(t_i)$};
  \node[title] (TB) at ($(B.west |- top)+(-0.1,0)$) {\textbf{(b)} Subtrajectories $\vx^{(j,k,c)}(\tau)$ for one $\alpha_k$};
  \begin{scope}[on background layer]
    \fill[panel] ($(A.west |- bot)+(-0.3,0)$) rectangle ($(A.east |- TA.north)+(0.3,0.15)$);
    \fill[panel] ($(B.west |- bot)+(-0.3,0)$) rectangle ($(B.east |- TB.north)+(0.3,0.15)$);
  \end{scope}
  \draw[link, line width=1.2pt, black!40] ($(A.east)+(0.4,0.9)$) -- node[above=1pt, note, align=center] {shift \&\\[-2pt] rescale} ($(B.west |- A.east)+(-0.4,0.9)$);

  \node[box=c1, anchor=north] (X0) at ($(A.south)+(0,-0.85)$)
    {\textbf{Normalized checkpoints}\\[3pt]
     $\mX_0^{(j,k)} = \Bigl[\,\dfrac{\vx^{(j,k)}((c-1)\Delta\tau)}{\lVert\vx^{(j,k)}((c-1)\Delta\tau)\rVert}\,\Bigr]_{c=1}^{N_c}$};
  \node[box=c2, anchor=north] (G) at ($(B.south |- X0.north)$)
    {\textbf{Gram matrix of windowed outputs}\\[3pt]
     $\mG^{(k)}_{\hat{i},\hat{j}} = \dfrac{1}{\alpha_k^2}\displaystyle\int_0^{\Delta T}\vy^{(\hat{i},k)}(\tau)^\top\vy^{(\hat{j},k)}(\tau)\,d\tau$};
  \draw[link] (X0.north |- bot) -- node[right, note] {normalize} (X0.north);
  \draw[link] (G.north |- bot) -- node[right, note] {integrate outputs} (G.north);
\end{tikzpicture}
  \caption{Schematic of the checkpointing procedure. (a) Impulse responses are stored for several amplitudes $\alpha_k$ (layers, ordered by $\lvert\alpha_k\rvert$, with the front layer corresponding to $\alpha_1$), and the procedure is applied separately to each amplitude level. Snapshots of each stored impulse response at the checkpoints $t = (c-1)\Delta\tau$ are normalized and used as initial conditions $\mX_0^{(j,k)}$ for the computation of the empirical observability Gramian. (b) The portion of the stored trajectory that follows each checkpoint is shifted in time and rescaled to have initial norm $\lvert\alpha_k\rvert$, yielding the subtrajectories \eqref{eq:windowing}, whose outputs define the Gram matrix $\mG^{(k)}$ in \eqref{eq:G_checkpoint}.}
  \label{fig:checkpointing}
\end{figure}
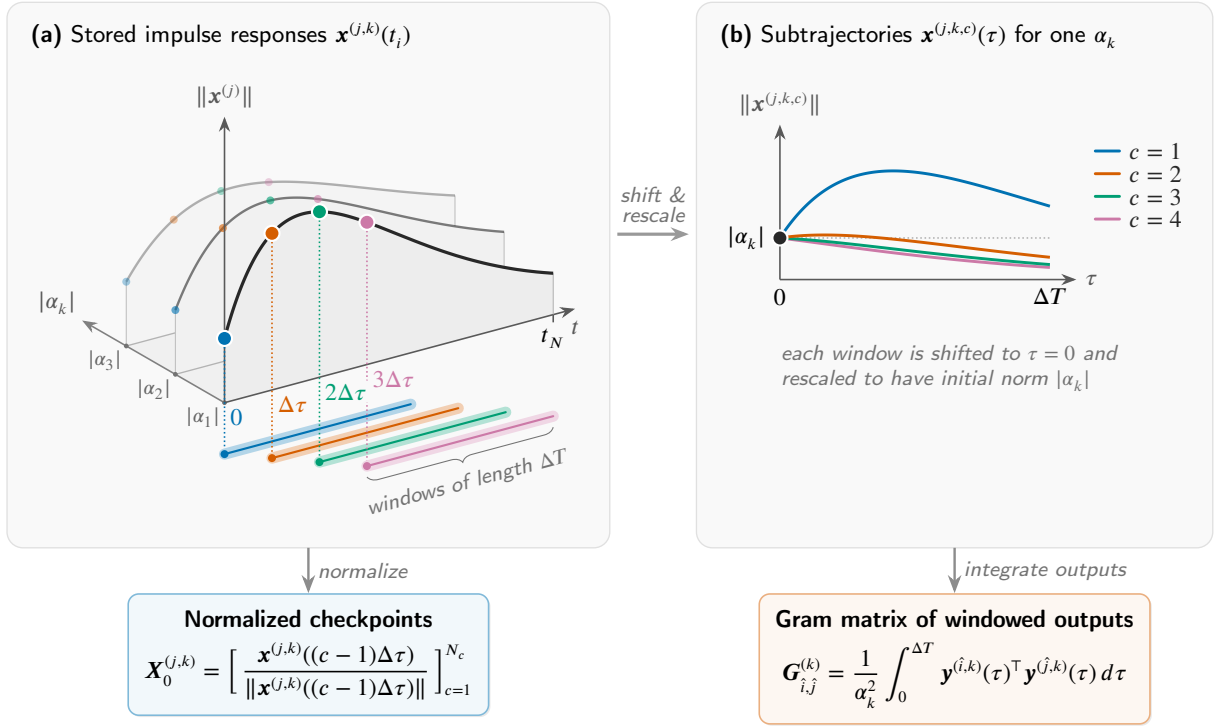
We do remark that the scaled trajectory $\vx^{(j,k,c)}(t)$ is, in general, not a solution of the FOM in equation \eqref{eq:nonlinear_fom}, unless the dynamics are linear, so, in principle, this procedure is only an approximation of the empirical observability Gramian.
Should this approximation prove insufficient, one can always run additional FOM simulations from the checkpointed initial conditions $\vx^{(j,k,c)}(0)$ to compute the empirical observability Gramian more accurately, but we do not find this necessary in the examples considered in this manuscript.
Defining the mapping $\hat{j} = (j-1)N_c + c$ so that $\vy^{(\hat{j},k)}(t) = \vy^{(j,k,c)}(t)$, we can then define the empirical observability Gramian as in \eqref{eq:empirical_WO}, but with
\begin{equation}
  \label{eq:G_checkpoint}
  \mG^{(k)}_{\hat{i},\hat{j}} = \frac{1}{\alpha_k^2} \int_0^{\Delta T} \vy^{(\hat{i},k)}(t)^\top \vy^{(\hat{j},k)}(t) dt,\quad \hat{i},\hat{j} \in \{1,2,\ldots,m\,N_c\}.
\end{equation}
Finally, the factorization $\widetilde{\mW}_O = \widetilde{\mY}\widetilde{\mY}^\top$ can be computed numerically by defining
\begin{equation}
  \label{eq:Y_factor}
  \widetilde{\mY} = s^{-1/2}\begin{bmatrix}
    \mY^{(1)} & \mY^{(2)} & \cdots & \mY^{(s)}
  \end{bmatrix}\in\R^{n\times s\,m\,N_c},\quad \mY^{(k)} = \mX_0^{(k)}\left(\mX_0^{(k)\top}\mX_0^{(k)}\right)^{-1}\mL^{(k)},
\end{equation}
where $\mL^{(k)}$ is the Cholesky factor of $\mG^{(k)}$, once the latter is evaluated numerically via quadrature.
The complete procedure is summarized in Algorithm~\ref{alg:checkpointed_balancing}.
We stress once more that the FOM is queried only in step~\ref{alg:step_fom}, and that all the remaining steps operate on the stored snapshots and outputs.
We close by observing that if the checkpointing procedure produces a rank-deficient $\mX_0^{(k)}$ matrix, then the orthogonal projector onto the range of $\mX_0^{(k)}$ can be defined using the pseudoinverse $\left(\mX_0^{(k)\top}\mX_0^{(k)}\right)^{+}$ obtained via a truncated SVD.
Similarly, if $\mG^{(k)}$ is rank-deficient or poorly conditioned, then the Cholesky factorization in step~\ref{alg:step_cholesky} can be replaced by a truncated SVD.
Moving forward, we drop the tilde notation, as it is understood that balanced test and trial bases are computed from the projected empirical observability Gramian $\widetilde{\mW}_O$.

\begin{algorithm}[t]
  \caption{Non-intrusive balancing via checkpointing}
  \label{alg:checkpointed_balancing}
  \begin{algorithmic}[1]
    \Require Amplitudes $\{\alpha_k\}_{k=1}^s$, sampling instances $\{t_i\}_{i=1}^N$ with quadrature weights $\{\gamma_i\}_{i=1}^N$, checkpoint spacing $\Delta\tau$, number of checkpoints $N_c$, rank $r$
    \Ensure Trial and test matrices $\widetilde{\mPhi},\widetilde{\mPsi}\in\R^{n\times r}$
    \For{$k = 1,\ldots,s$ and $j = 1,\ldots,m$}
      \State \label{alg:step_fom} Evolve \eqref{eq:nonlinear_fom} from $\vx^{(j,k)}(0) = \alpha_k\mB_j$ and store $\vx^{(j,k)}(t_i)$ and $\vy^{(j,k)}(t_i)$ for $i = 1,\ldots,N$
    \EndFor
    \State Assemble the reachability factor $\mX$ via \eqref{eq:X_factor}--\eqref{eq:X_blocks}
    \For{$k = 1,\ldots,s$}
      \State Assemble the checkpoints $\mX_0^{(k)}$ via \eqref{eq:checkpoints}
      \State Window and scale the stored outputs via \eqref{eq:windowing}
      \State \label{alg:step_cholesky} Compute $\mG^{(k)}$ via \eqref{eq:G_checkpoint} by quadrature, and its Cholesky factorization $\mG^{(k)} = \mL^{(k)}\mL^{(k)\top}$
    \EndFor
    \State Assemble the observability factor $\widetilde{\mY}$ via \eqref{eq:Y_factor}
    \State Compute the singular value decomposition $\widetilde{\mY}^\top\mX = \widetilde{\mU}\widetilde{\mSigma}\widetilde{\mV}^\top$
    \State Compute $\widetilde{\mPhi}$ and $\widetilde{\mPsi}$ via \eqref{eq:PhiPsi}
  \end{algorithmic}
\end{algorithm}

\subsection{Learning Latent-Space Dynamics}
\label{subsec:learning_lsd}

In this section, we briefly summarize the Operator Inference procedure for learning latent-space dynamics from high-fidelity data \citep{peherstorfer_data-driven_2016,kramer_learning_2024}.
This is our method of choice for learning the latent dynamics because, under a polynomial ansatz, the solution can be computed in closed form without the need for gradient-based optimization, as we shall see momentarily.
However, we once again emphasize that any other method (e.g., Neural ODEs \citep{chen2018neuralode}, SINDy \citep{brunton2016sindy}, Gaussian Processes \citep{rasmussen2006gp}, etc.) can be used instead.
Once the encoder $\mPsi^\top$ and decoder $\mPhi\left(\mPsi^\top\mPhi\right)^{-1}$ are available, Operator Inference projects the high-dimensional snapshots $\vx(t_i)$ from all training trajectories onto the latent space to obtain the reduced coordinates $\vz(t_i) = \mPsi^\top\vx(t_i)$.
Assuming that the latent-space dynamics $\vf_r(\vz,\vu)$ can be taken as polynomial functions of the reduced coordinates $\vz$ and the inputs $\vu$, e.g., $\vf_r(\vz,\vu) = \mA_r\vz + \mH_r:\left(\vz\otimes\vz\right) + \mB_r\vu$, the reduced-order tensors $\mA_r \in \mathbb{R}^{r\times r}$, $\mH_r \in \mathbb{R}^{r\times r\times r}$, and $\mB_r \in \mathbb{R}^{r\times m}$ can be inferred by solving the following least-squares problem
\begin{equation}
  \label{eq:operator_inference}
  \min_{\left(\mA_r,\mH_r,\mB_r\right)} \sum_{i} \gamma_i \left\lVert \dot{\vz}(t_i) - \vf_r\left(\vz(t_i),\vu(t_i)\right) \right\rVert_2^2.
\end{equation}
Here, $\gamma_i$ are weights accounting for the importance of each data point (e.g., if the snapshots are sampled non-uniformly in time), and $\dot{\vz}(t_i)$ is the time derivative of the reduced coordinates, which can be computed via finite differences from the stored snapshots $\vz(t_i)$.
Alternatively, it can be computed exactly by evaluating the full-order dynamics, i.e., $\dot{\vz}(t_i) = \mPsi^\top\vf(\vx(t_i),\vu(t_i))$.
A third procedure, which we adopt in this work, is to perform a so-called ``reprojection,'' where $\dot{\vz}(t_i) = \mPsi^\top\vf(\mPhi\left(\mPsi^\top\mPhi\right)^{-1}\vz(t_i),\vu(t_i))$.
In words, $\dot{\vz}$ is obtained by encoding the full-order dynamics~$\vf$ evaluated at $\hat{\vx}\coloneqq\mathbb{P}\vx$, where $\mathbb{P} = \mPhi\left(\mPsi^\top\mPhi\right)^{-1}\mPsi^\top$ is the oblique projector associated with the encoder-decoder pair under consideration.
It was shown in \cite{peherstorfer_data-driven_2016} that, under mild conditions, this procedure exactly recovers the tensors that one would obtain via Petrov-Galerkin projection of the full-order dynamics $\vf$.

While the least-squares problem \eqref{eq:operator_inference} can be solved conveniently in closed form, it provides no guarantee that the inferred latent-space dynamics are stable, and the resulting models can exhibit finite-time blow-up or unbounded growth when evaluated outside of the training regime.
When stability is a concern, we therefore adopt the globally-asymptotically-stable Operator Inference (GasOpInf) formulation of \cite{goyal_guaranteed_2025}, with the minor modifications introduced in \cite{errico2026gasnitrom}.
For quadratic latent-space dynamics $\vf_r(\vz,\vu) = \mA_r\vz + \mH_r:\vz\vz^\top + \mB_r\vu$, global asymptotic stability of the unforced system is guaranteed if there exists a quadratic Lyapunov function $V(\vz) = \vz^\top\widetilde{\mQ}\vz$, with $\widetilde{\mQ}$ symmetric positive definite, such that the linear dynamics are strictly dissipative and the quadratic dynamics are energy-preserving with respect to $\widetilde{\mQ}$, i.e.,
\begin{equation}
  \label{eq:gas_conditions}
  \vz^\top\left(\widetilde{\mQ}\mA_r + \mA_r^\top\widetilde{\mQ}\right)\vz < 0,\qquad \vz^\top\widetilde{\mQ}\left(\mH_r:\vz\vz^\top\right) = 0,\qquad \forall\,\vz\neq\bm{0}.
\end{equation}
Following \cite{gillis_computing_2017} and \cite{goyal_guaranteed_2025}, these conditions are enforced by construction by parameterizing the reduced-order tensors as
\begin{equation}
  \label{eq:gas_parameterization}
  \mA_r = \left(\mK - \mK^\top - \mR^{-1}\mR^{-\top}\right)\mQ^{-1}\mQ^{-\top},\qquad \mH_{r,k} = \left(\mS_k - \mS_k^\top\right)\mQ^{-1}\mQ^{-\top},
\end{equation}
where $\widetilde{\mQ} = \mQ^{-1}\mQ^{-\top}$, $\mH_{r,k}$ denotes the $k$th frontal slice of $\mH_r$, and $\mK$, $\mR$, $\mQ\in\R^{r\times r}$ and $\mS_k\in\R^{r\times r}$ are unconstrained matrices (with $\mR$ and $\mQ$ invertible).
GasOpInf then minimizes the same objective function as in \eqref{eq:operator_inference}, but over the unconstrained parameters $(\mK,\mR,\mQ,\{\mS_k\}_{k=1}^r,\mB_r)$.
Since the objective function is no longer quadratic in these parameters, the solution must be computed iteratively using a gradient-based algorithm, and we refer the reader to \cite{errico2026gasnitrom} for closed-form expressions of the gradients and for a convenient initialization strategy.

\section{Results}
\label{sec:results}

In this section, we demonstrate the performance of the proposed data-driven balancing procedure on several benchmark problems of increasing size and complexity, and we compare against standard approaches based on orthogonal projections onto POD modes. 
Specifically, our models are OpInf or GasOpInf models learned from the reduced coordinates obtained via the balancing procedure in Algorithm~\ref{alg:checkpointed_balancing}, and we compare against equivalent OpInf or GasOpInf models learned from POD modes obtained from the same training data.
All instances of OpInf and GasOpInf compute the time derivatives of the reduced coordinates via reprojection, as discussed in Section~\ref{subsec:learning_lsd}, as we have found that this promotes stability and translates into improved accuracy.
The examples include the complex Ginzburg-Landau equation at two different parameter values (one for which the origin is stable and one for which it is unstable), the flow past a NACA 0012 airfoil at Reynolds number $Re = 600$ and angle of attack $6^\circ$, and the flow past an aspect-ratio-one wing with NACA 0012 profile at Reynolds number $Re = 1000$ and angle of attack $6^\circ$.

\subsection{Complex Ginzburg-Landau Equation with a Stable Origin}
\label{subsec:cgl_stable}

Our first example is the complex Ginzburg-Landau (CGL) equation
\begin{equation}
  \label{eq:cgl}
  \frac{\partial q}{\partial t} = \left(-\nu\frac{\partial}{\partial x} + \gamma\frac{\partial^2}{\partial x^2} + \mu(x)\right) q - a\lvert q\rvert^2 q,\quad x\in(-\infty,\infty),\quad q(x,t)\in\mathbb{C},
\end{equation}
which is a widely used model for the spatio-temporal evolution of disturbances in spatially developing shear flows \citep{ilak2010cgl}.
We adopt the parameters in \cite{ilak2010cgl}, namely $a = 0.1$, $\gamma = 1 - i$, $\nu = 2 + 0.4i$, and a spatially varying amplification rate $\mu(x) = \mu_0 - c_u^2 + \mu_2 x^2/2$ with $c_u = 0.2$, $\mu_2 = -0.01$ and $\mu_0 = 0.38$.
For this choice of parameters, the origin $q = 0$ is linearly stable, but the linearized dynamics are highly non-normal and, consequently, small perturbations can undergo large transient amplification before eventually decaying.
This makes the CGL equation a challenging benchmark for model reduction techniques that rely on orthogonal projections, and a natural benchmark for balancing-based techniques.

The system is actuated and observed via spatially localized Gaussian kernels.
Specifically, the input $u(t)\in\mathbb{C}$ and the output $y(t)\in\mathbb{C}$ are defined as
\begin{equation}
  \label{eq:cgl_io}
  (\tilde\mB u)(x,t) = \exp\left\{-\left(\frac{x - \bar{x}}{s}\right)^2\right\}u(t),\qquad y(t) = (\mC q)(t) = \int_{-\infty}^{\infty}\exp\left\{-\left(\frac{x + \bar{x}}{s}\right)^2\right\}q(x,t)\,dx,
\end{equation}
with $s = 1.6$ and $\bar{x} = -\sqrt{-2(\mu_0 - c_u^2)/\mu_2}$, so that the actuator is located at the upstream boundary (``branch I'') of the region of local instability and the sensor is located at its downstream mirror image \citep{ilak2010cgl}.
Upon spatial discretization with a fourth-order finite difference scheme on a grid with $n=256$ nodes, and upon splitting the complex-valued quantities into their real and imaginary parts, equation \eqref{eq:cgl} can be written as a real-valued system with cubic dynamics
\begin{equation}
  \label{eq:cgl_discrete}
  \frac{d}{dt}\vq = \mA\vq + \mH:\left(\vq\otimes\vq\otimes\vq\right) + \mB\vu,\qquad \vy = \mC\vq,
\end{equation}
where $\vq\in\R^{2n}$, $\vu\in\R^2$, $\vy\in\R^2$, and $\mB_j = \tilde{\mB}_j/\lVert\tilde{\mB}_j\rVert$, where $\lVert \cdot\rVert$ is the Euclidean norm on $\mathbb{R}^{2n}$.
In the notation of Section~\ref{sec:math_formulation}, this system has $m = p = 2$ and it is control-affine, so that Algorithm~\ref{alg:checkpointed_balancing} applies directly.
Given the cubic nature of the full-order dynamics, we seek reduced-order models with linear, cubic and input terms, i.e., $\vf_r(\vz,\vu) = \mA_r\vz + \mH_r:\left(\vz\otimes\vz\otimes\vz\right) + \mB_r\vu$.
(The input tensor $\mB_r$ is not inferred, but rather it is computed via Petrov-Galerkin projection of the full-order input matrix $\mB$ onto the latent space, i.e., $\mB_r = \mPsi^\top\mB$.)
The training and testing data are generated as in \cite{padovan2024nitrom}.
Specifically, we train on the impulse responses of \eqref{eq:cgl_discrete} from the initial conditions $\vq^{(j,k)}(0) = \alpha_k\mB_j$, with $j\in\{1,2\}$ and $\alpha_k\in\{-1.0,\,0.01,\,0.1,\,1.0\}$, for a total of $s\,m = 8$ training trajectories.
The state and output along each trajectory are sampled at $N = 2501$ uniformly spaced time instances in $[0,500]$, with $\Delta t = 0.2$.
These are the only FOM simulations required by Algorithm~\ref{alg:checkpointed_balancing}.
In order to perform the balancing procedure, we extract $N_c = 30$ checkpoints from each trajectory, spaced $\Delta\tau = 0.6$ time units apart, and we compute $r=6$-dimensional trial and test subspaces using Algorithm~\ref{alg:checkpointed_balancing}. 
At $r=6$, we retain more than $99\%$ of the energy of the Hankel singular values of the projected Gramians.

We compare the resulting Balanced Operator Inference (BOpInf) model with a standard (i.e., POD-based) Operator Inference (OpInf) model of the same dimension, in which the full-order data are projected orthogonally onto the span of the leading POD modes of the training data.
For this example, $r=6$ was sufficient to capture more than $99\%$ of the variance of the training data.
Both models are tested on impulse responses whose amplitudes are not included in the training set, and on sinusoidal inputs that were never seen during training.
In particular, we generate $N_{\text{test}} = 50$ testing impulse responses from initial conditions $\alpha\mB_l$, with $l\in\{1,2\}$ and with $\alpha$ drawn uniformly at random from $[-1,1]$, and we define the testing error along the $j$th testing trajectory as
\begin{equation}
  \label{eq:cgl_error}
  e_j(t) = \frac{\lVert\vy^{(j)}(t) - \hat{\vy}^{(j)}(t)\rVert^2}{\eta_j},\quad \eta_j = \frac{1}{t_f}\int_0^{t_f}\lVert\vy^{(j)}(t)\rVert^2 dt,
\end{equation}
where $\vy^{(j)}$ and $\hat{\vy}^{(j)}$ are the outputs of the FOM and of the ROM, and $\eta_j$ is the time-averaged energy of $\vy^{(j)}$.
The results for this example are collected in Figure~\ref{fig:cgl_responses}, where the upper and lower panels of each pair correspond to $\mu_0 = 0.38$ and $\mu_0 = 0.41$, respectively.
The errors along all testing trajectories are shown in the top-left panels.
Except for a brief initial transient, the error of the BOpInf model remains at the level of a few percent over the whole time horizon for all testing trajectories, which is between one and three orders of magnitude below that of the OpInf model.
By contrast, while the OpInf model is accurate at very early times, its error rapidly grows to $\mathcal{O}(1)$ as the disturbance undergoes transient amplification, and it eventually exceeds $\mathcal{O}(10)$ for most testing trajectories as the model develops spurious oscillations at late times.
These observations are confirmed by the representative testing impulse response with amplitude $\alpha = 0.3$ in the top-right panels: the BOpInf model tracks the large-amplitude transient growth and the subsequent decay of the output, while the OpInf model fails to capture the transient and predicts a spurious late-time response.
Finally, the bottom panels of Figure~\ref{fig:cgl_responses} show the response to the sinusoidal inputs $w(t) = 0.05\sin(\omega t)$ and $w(t) = 0.2\sin(2\omega t)$, where $\omega\approx 0.648$ is the natural frequency of the system.
(We recall that we did not train on sinusoidal forcing inputs.)
As in \cite{padovan2024nitrom}, these inputs enter the dynamics as $\mB\vu(t) = w(t)\,\mB\vv/\lVert\mB\vv\rVert$, where $\vv\in\R^2$ is a randomly chosen vector, and the system is initialized from the origin.
At the natural frequency, the BOpInf model reproduces the amplitude and phase of the forced response of the FOM almost exactly, while the OpInf model underestimates the response for most of the time horizon and then drifts into a spurious oscillation.
At the second harmonic $2\omega$, which is a more demanding test since the training data are dominated by dynamics at the natural frequency, the BOpInf model closely tracks the FOM response over the whole time horizon, while the OpInf model develops spurious large-amplitude oscillations for $t\gtrsim 100$.
We attribute the improved performance of the proposed method to the fact that it achieves dimensionality reduction via an oblique projection informed by balancing.
To illustrate this point, in Figure~\ref{fig:cgl_modes} we compare the leading POD modes with the columns of the trial and test matrices $\mPhi$ and $\mPsi$ computed by Algorithm~\ref{alg:checkpointed_balancing}.
The POD modes, which capture the most energetic structures in the training data, are spread over the region of local instability between the actuator and the sensor.
By contrast, the (balanced) trial modes are concentrated near the sensor at $x = -\bar{x}$, where the disturbances reach their largest amplitude and where they are observed, while the test modes (which define the encoder) are concentrated near the actuator at $x = \bar{x}$, where the system is most sensitive to perturbations and where the adjoint modes of the linearized dynamics are typically supported (see, e.g., Figure~6 in~\cite{CHEN_ROWLEY_2011}).
This spatial separation between the trial and test modes is a hallmark of non-normal dynamics, and it cannot be represented by an orthogonal projection, where the trial and test subspaces coincide.
Finally, we note that the modes in Figure~\ref{fig:cgl_modes} come in pairs with identical magnitude because the CGL equation and its inputs and outputs are equivariant under the phase rotation $q\mapsto e^{i\theta}q$: consequently, the Gramians commute with multiplication by the imaginary unit, so that if $\phi(x)$ is a mode, then so is $i\phi(x)$, with the same singular value.

\begin{figure}
  \centering
  \includegraphics[width=0.75\textwidth]{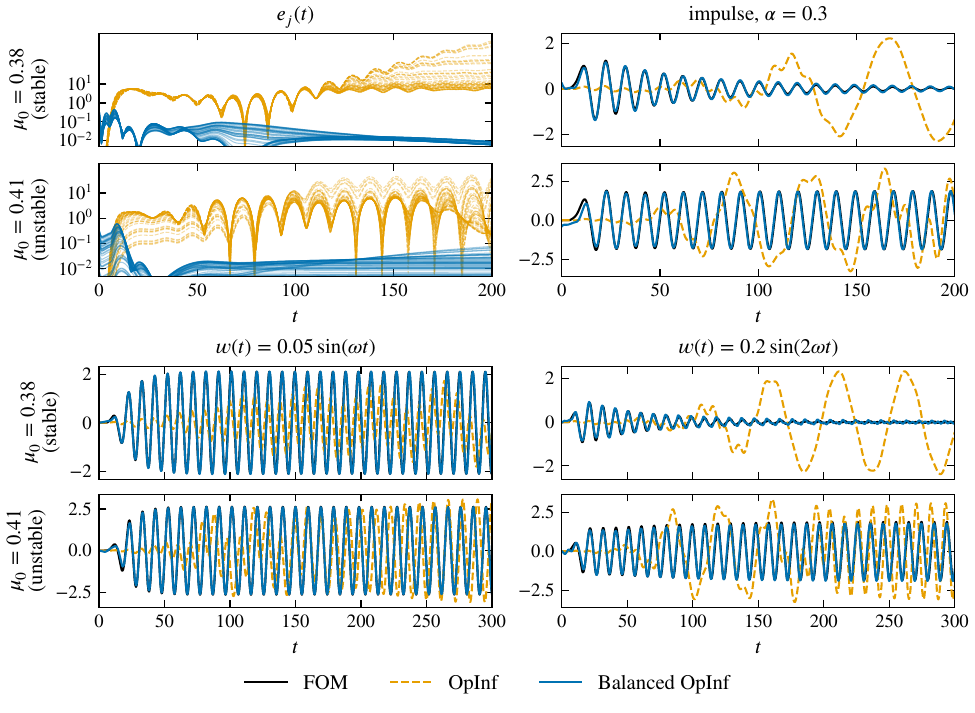}
  \caption{CGL equation. Testing errors \eqref{eq:cgl_error} along all testing trajectories (top left), real part of the output from a representative testing impulse response with amplitude $\alpha = 0.3$ (top right), and real part of the output in response to sinusoidal inputs at frequencies $\omega$ and $2\omega$ (bottom). In each pair of panels, the upper panel corresponds to the stable case ($\mu_0 = 0.38$, Section~\ref{subsec:cgl_stable}) and the lower panel to the unstable case ($\mu_0 = 0.41$, Section~\ref{subsec:cgl_unstable}).}
  \label{fig:cgl_responses}
\end{figure}

\begin{figure}
  \centering
  \includegraphics[width=0.75\textwidth]{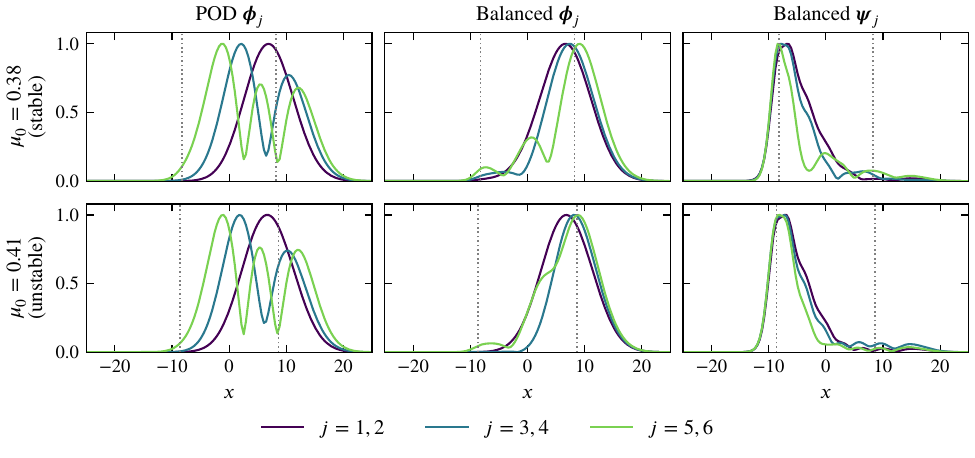}
  \caption{CGL equation. Magnitude of the leading POD modes (left), and of the columns $\vphi_j$ and $\vpsi_j$ of the trial and test matrices $\mPhi$ and $\mPsi$ computed by Algorithm~\ref{alg:checkpointed_balancing} (middle and right), for the stable (top) and unstable (bottom) cases. Each mode is normalized to unit maximum magnitude, and the vertical dotted lines indicate the locations of the actuator ($x = \bar{x}$) and of the sensor ($x = -\bar{x}$).}
  \label{fig:cgl_modes}
\end{figure}

\FloatBarrier
\subsection{Complex Ginzburg-Landau Equation with an Unstable Origin}
\label{subsec:cgl_unstable}

We now consider the CGL equation \eqref{eq:cgl} with $\mu_0 = 0.41$, while all the remaining parameters, the inputs and outputs, the spatial discretization, and the training and testing procedures are identical to those described in Section~\ref{subsec:cgl_stable}.
In particular, we train on $s\,m = 8$ impulse responses sampled at $N = 2501$ time instances in $[0,500]$, we extract $N_c = 30$ checkpoints spaced $\Delta\tau = 0.6$ apart from each trajectory, and we compute models of dimension $r = 6$.
For this value of $\mu_0$, the origin is linearly unstable, and small perturbations grow exponentially until they saturate onto a stable time-periodic limit cycle \citep{ilak2010cgl}.
We include this example to demonstrate that the proposed procedure does not require the existence of a neighboring stable steady state.
In fact, while the infinite-time Gramians associated with the (unstable) linearization about the origin are not defined, the empirical Gramians in Section~\ref{subsubsec:data_driven_nonlinear_balancing} are computed from finite-time trajectories sampled from the nonlinear system, and therefore remain well defined.

The testing errors and a representative testing impulse response are shown in the top panels of Figure~\ref{fig:cgl_responses} (lower panel of each pair).
Once again, after a brief initial transient, the error of the BOpInf model drops to a few percent and remains at this level over the entire time horizon, approximately two orders of magnitude below that of the OpInf model.
In particular, the impulse response shows that the BOpInf model locks onto the limit cycle within the first couple of oscillation periods, and it predicts its amplitude and frequency with high accuracy thereafter.
The OpInf model, on the other hand, fails to capture both the transient and the limit cycle, and it develops an irregular oscillation whose amplitude is up to approximately $50\%$ larger than that of the FOM.
Similar conclusions can be drawn from the responses to the sinusoidal inputs shown in the bottom panels of Figure~\ref{fig:cgl_responses}.
For both forcing frequencies, the BOpInf model closely tracks the amplitude and phase of the FOM response, while the OpInf model eventually settles onto an oscillation with incorrect amplitude.
Finally, the POD modes and the trial and test modes are shown in the bottom row of Figure~\ref{fig:cgl_modes}.
As in the stable case, the trial modes are localized near the sensor and the test modes are localized near the actuator, indicating that the balancing procedure identifies an oblique projection that accounts for the non-normality of the dynamics even when the latter evolve on a limit cycle far from the origin.

\FloatBarrier
\subsection{Flow Over an Airfoil}
\label{subsec:airfoil}

Our next example is the two-dimensional incompressible flow over a NACA 0012 airfoil at an angle of attack of $6^\circ$ and at a chord-based Reynolds number $Re = 600$.
The flow is governed by the incompressible Navier-Stokes equations, which we simulate with an in-house solver that implements the immersed boundary projection method of \cite{Taira2007} with a second-order finite-volume discretization on a fully staggered grid.
The (non-dimensional) computational domain is $[-5,16]\times[-4,4]$, with the mid-chord of the airfoil located at the origin.
A uniform velocity $(u,v) = (1,0)$ is imposed at all boundaries except for the outflow, where we impose a homogeneous Neumann condition on the velocity.
For this choice of parameters, the flow admits a linearly stable steady state, whose vorticity field is shown in Figure~\ref{fig:airfoil_base_modes}(a).
However, the linearized dynamics about this steady state are highly non-normal, and small perturbations near the surface of the airfoil are amplified by several orders of magnitude as they are advected into the wake, where they develop into vortex-shedding-like wavepackets before eventually decaying.

\begin{figure}
  \centering
  \includegraphics[width=0.65\textwidth]{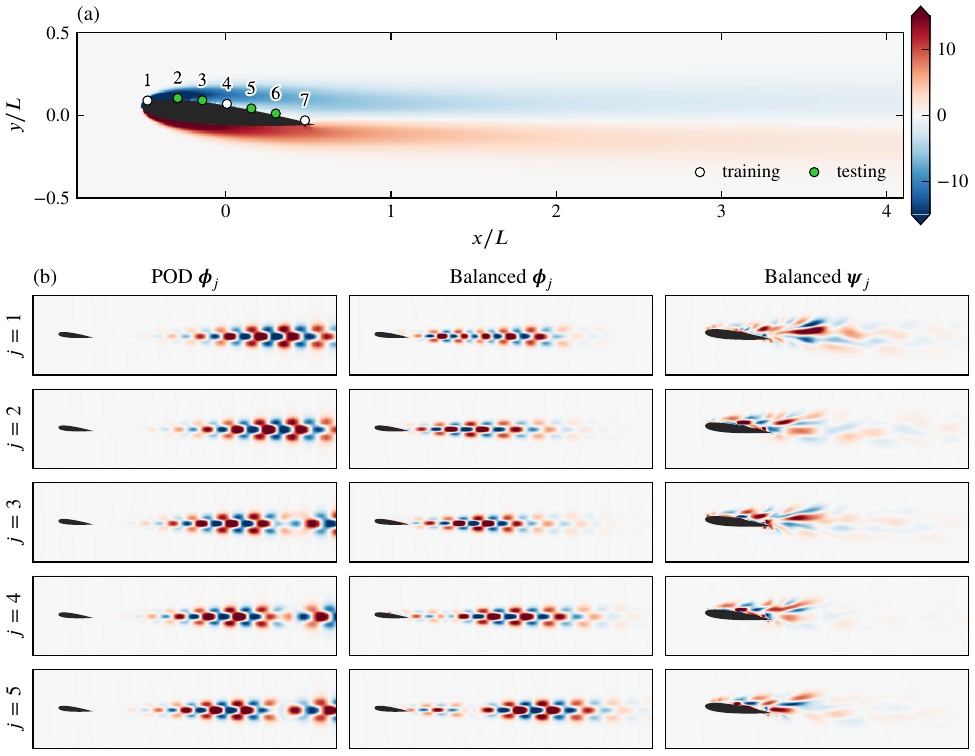}
  \caption{Airfoil flow. (a) Vorticity of the steady state, with the forcing stations used for training (white) and testing (green). (b) Vorticity of the leading POD modes (left), and of the columns $\vphi_j$ and $\vpsi_j$ of the trial and test matrices $\mPhi$ and $\mPsi$ computed by Algorithm~\ref{alg:checkpointed_balancing} (middle and right).}
  \label{fig:airfoil_base_modes}
\end{figure}

We consider the dynamics of perturbations $\vq(t)$ about the steady state, forced by spatially localized inputs
\begin{equation}
  \label{eq:airfoil_forcing}
  b(x,y)\,w(t) = \exp\left\{-1250\left((x - x_c)^2 + (y - y_c)^2\right)\right\}w(t),
\end{equation}
where $(x_c, y_c)$ is located along the upper surface of the airfoil at one of the seven stations shown in Figure~\ref{fig:airfoil_base_modes}(a).
Upon spatial discretization and removal of the pressure and of the immersed-boundary forces via projection onto the space of discretely divergence-free velocity fields, the perturbation dynamics take the control-affine quadratic form
\begin{equation}
  \label{eq:airfoil_fom}
  \frac{d}{dt}\vq = \mA\vq + \mH:\vq\vq^\top + \mB\vu,
\end{equation}
where the $j$th column of $\mB$ is the (Leray-projected and unit-norm) discretization of \eqref{eq:airfoil_forcing} at the $j$th station.
Since the computational mesh is highly stretched away from the airfoil, we work with the weighted state $\widetilde{\vq} = \mW^{1/2}\vq$, where $\mW$ is a diagonal matrix containing the ratio between the volume of each cell and the volume of the smallest cell in the computational domain, so that the Euclidean norm of $\widetilde{\vq}$ is a (scaled) measure of the kinetic energy of the perturbation.
Throughout this section, we measure the whole weighted state, i.e., $\vy = \widetilde{\vq}$, and the state dimension is $n = 2\,n_x\, n_y = 6.75\times 10^5$, with $n_x = 750$ and $n_y = 450$ grid points in the streamwise and wall-normal directions, respectively.

\paragraph{Training.}
We train on impulse responses forced at stations~1, 4 and 7 in Figure~\ref{fig:airfoil_base_modes}(a), which are located near the leading edge, the mid-chord and the trailing edge of the airfoil, respectively.
That is, we simulate \eqref{eq:airfoil_fom} from the initial conditions $\vq^{(j,k)}(0) = \alpha_k\mB_j$, with $j\in\{1,2,3\}$ indexing the training stations and $\alpha_k\in\{0.1,\,1.0,\,2.0\}$, for a total of $s\,m = 9$ training trajectories.
Since the impulse responses evolve on a fast time scale right after the impulse, we store snapshots every $0.02$ time units for $t\in[0,4]$, and every $0.2$ time units for $t\in(4,39.8]$.
We stress that both the balancing procedure and POD are applied to the weighted snapshots $\widetilde{\vq}(t_i) = \mW^{1/2}\vq(t_i)$, so that all inner products (and hence all Gramians and POD covariances) are computed in the energy-weighted norm described above.
The empirical reachability Gramian is then factored as in \eqref{eq:X_factor}--\eqref{eq:X_blocks}, with trapezoidal quadrature weights $\gamma_i = (t_{i+1} - t_{i-1})/2$ to account for the non-uniform sampling.
The trial and test matrices $\mPhi$ and $\mPsi$ are computed using Algorithm~\ref{alg:checkpointed_balancing}, with a few modifications to the checkpointing procedure of Section~\ref{subsubsec:computational_algorithm} that account for the multiscale nature of the impulse responses.
First, rather than using uniformly spaced checkpoints, we place $20$ checkpoints at $t\in\{0, 0.1, \ldots, 1.9\}$ to resolve the fast transient that follows the impulse, and $19$ additional checkpoints approximately uniformly spaced on $[0.2, 20]$ (each moved to the nearest stored snapshot), for a total of $N_c = 39$ checkpoints per trajectory.
Second, all subtrajectories in \eqref{eq:windowing} are defined over a window of length $\Delta T = 19.8$, which is the time between the last checkpoint and the end of the stored trajectory, and the integrals in \eqref{eq:G_checkpoint} are evaluated with the rectangle rule on a uniform grid with spacing $0.2$.
Subtrajectory samples that do not coincide with a stored snapshot are computed via six-point local Lagrange interpolation.
Finally, given that consecutive checkpoints along a trajectory can be nearly collinear, the inverse $\left(\mX_0^{(k)\top}\mX_0^{(k)}\right)^{-1}$ in \eqref{eq:Y_factor} is replaced by a pseudoinverse in which singular values below $10^{-4}$ are truncated, as discussed at the end of Section~\ref{subsubsec:computational_algorithm}.
We take $r = 50$, which is sufficient for both the leading POD modes and the balanced modes to capture more than $99\%$ of the energy in their respective spectra.
The reduced-order dynamics are then inferred using GasOpInf (see Section~\ref{subsec:learning_lsd}), so that the inferred quadratic dynamics are globally asymptotically stable by design.
As in the CGL example, the input tensor is not inferred, but rather it is computed via Petrov-Galerkin projection of the (weighted) full-order input matrix, i.e., $\mB_r = \mPsi^\top\mW^{1/2}\mB$, so that GasOpInf only optimizes over the parameters that define $\mA_r$ and $\mH_r$.
The wall-clock times for this example are reported in Table~\ref{tab:timings_airfoil}, which shows that the cost of computing the balanced bases is comparable to that of computing POD modes, and that both are small compared to the cost of generating the training data.

\begin{table}
  \centering
  \caption{Airfoil flow. Wall-clock times (in seconds). The FOM was simulated on an NVIDIA GH200 GPU, and the ROM time refers to the integration of one trajectory over the training horizon $t\in[0,40]$. The GasOpInf fit refers to the balanced basis (the fit in the POD basis took approximately $98$ seconds).}
  \label{tab:timings_airfoil}
  \begin{tabular}{ccccc}
    \toprule
    POD & Balancing (Alg.~\ref{alg:checkpointed_balancing}) & GasOpInf fit & ROM integration & FOM integration \\
    \midrule
    $20.6$ & $16.2$ & $24.2$ & $1.09$ & $255$ \\
    \bottomrule
  \end{tabular}
\end{table}

We refer to the resulting model as Balanced GasOpInf (BGasOpInf), and we compare it with a GasOpInf model of the same dimension trained on the same data, but where the high-dimensional data are projected orthogonally onto the span of the leading POD modes.
Since the snapshots are sampled non-uniformly in time, the POD modes are computed from the same non-uniformly sampled training data used for balancing, with each weighted snapshot $\widetilde{\vq}(t_i)$ scaled by $\sqrt{\gamma_i}$, where $\gamma_i$ are the same trapezoidal quadrature weights used for the empirical reachability Gramian, so that the POD modes are representative of the time-integrated (rather than sample-averaged) energy of the data.
The leading $5$ POD modes and the columns of $\mPhi$ and $\mPsi$ are shown in Figure~\ref{fig:airfoil_base_modes}(b).
The POD modes are dominated by large-amplitude wavepackets far downstream in the wake, where the training trajectories carry most of their energy.
The trial modes $\vphi_j$ display similar wavepackets, but they are concentrated closer to the airfoil, where the perturbations are still growing.
The trial modes $\vphi_j$ with $j > 5$ (not shown) are supported farther downstream, and resemble the leading POD modes.
The test modes $\vpsi_j$, on the other hand, are localized near the surface of the airfoil and in the near wake, which is where the flow is most receptive to perturbations and where the adjoint modes of the linearized dynamics are typically supported.
As in the CGL examples, the spatial separation between the trial and test modes is a signature of the non-normality of the underlying dynamics, and it is precisely what an orthogonal projection onto POD modes fails to capture.

\begin{figure}
  \centering
  \includegraphics[width=0.85\textwidth]{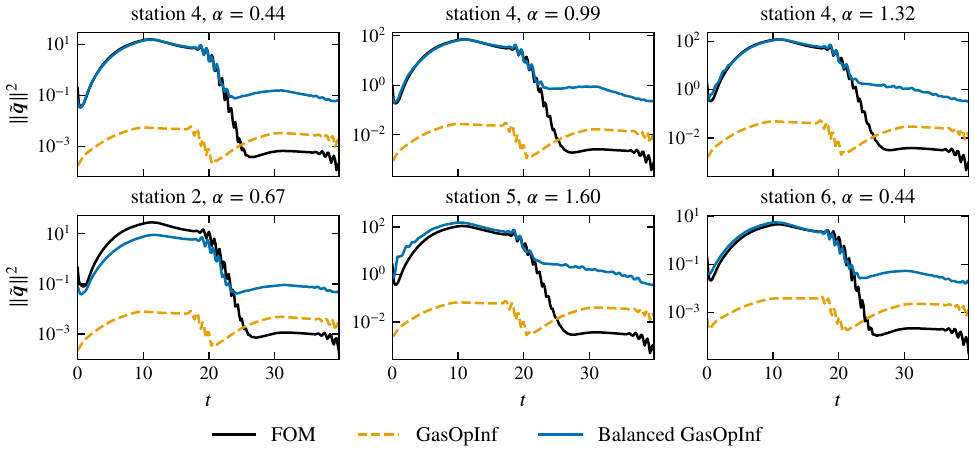}
  \caption{Airfoil flow. Perturbation energy of six representative testing impulse responses forced at the training station~4 (top) and at the unseen stations~2, 5 and~6 (bottom).
  }
  \label{fig:airfoil_testing}
\end{figure}

\begin{figure}
  \centering
  \includegraphics[width=0.85\textwidth]{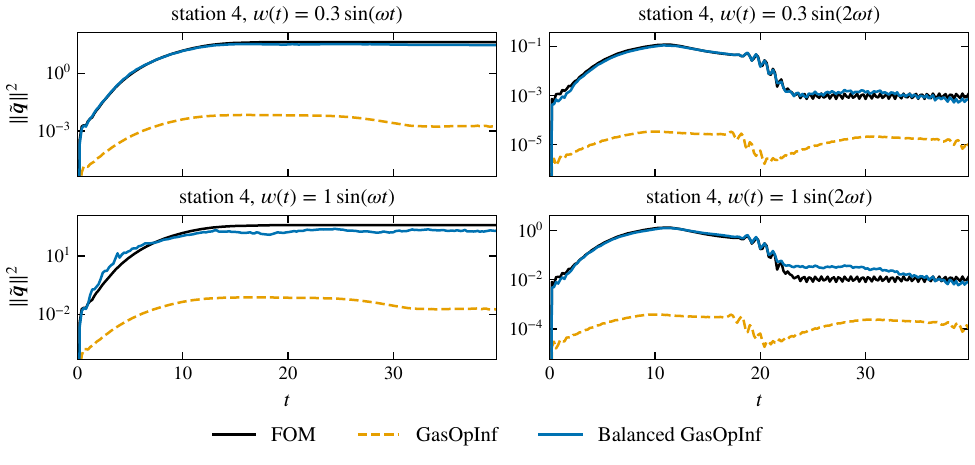}
  \caption{Airfoil flow. Perturbation energy in response to sinusoidal forcing at the mid-chord station~4, with amplitudes $a\in\{0.3,1\}$ (top and bottom) and frequencies $\omega$ and $2\omega$ (left and right).}
  \label{fig:airfoil_forced}
\end{figure}

\paragraph{Testing.}
We first test the models on impulse responses with amplitudes that were not seen during training, forced both at the mid-chord training station~4 and at the unseen stations~2, 3, 5 and~6.
The latter is a demanding test, since the corresponding input directions were never seen during training.
In all cases, the testing amplitudes are drawn uniformly at random from the training range $[0.1, 2]$.
The time history of the perturbation energy $\lVert\widetilde{\vq}\rVert^2$ for six representative testing trajectories is shown in Figure~\ref{fig:airfoil_testing}.
The GasOpInf model completely fails to capture the response: its predicted energy remains several orders of magnitude below that of the FOM throughout the transient, so that it is essentially no more accurate than the trivial prediction $\hat{\vq}(t)\equiv\bm{0}$.
By contrast, the BGasOpInf model accurately tracks the transient growth of the perturbation, including its peak around $t\approx 10$ and its subsequent decay, both at the training station~4 and at the unseen stations~5 and~6, where its predictions closely follow those of the FOM throughout the active part of the response.
Even at station~2, which is close to the leading edge and therefore the most challenging of the unseen stations, the model correctly captures the growth and decay of the response, with a peak energy that is within a factor of approximately three of the FOM (as opposed to several orders of magnitude for GasOpInf).
Once the perturbation has decayed by several orders of magnitude ($t\gtrsim 25$), the BGasOpInf model retains a small residual energy, which is nonetheless two to three orders of magnitude below the peak of the response.
We hypothesize that the model struggles to track the rapid decay of the perturbation at late times (due to vortical structures being advected out of the domain) because accommodating such behavior would require higher-than-quadratic dynamics in 50 dimensions.
For completeness, we also remark that the high-frequency oscillations at times $t\approx 20$ in Figure~\ref{fig:airfoil_testing} are due to small numerical artifacts as the wavepacket advects out of the computational domain through the homogeneous Neumann outflow boundary condition, and they have no measurable effect on our conclusions.
Vorticity snapshots from the testing impulse at the unseen station~5 with $\alpha = 1.60$ are shown in Figure~\ref{fig:airfoil_snapshots}(a), where we see that the BGasOpInf model correctly predicts the location, the shape and the amplitude of the wavepacket as it is advected downstream, while the GasOpInf model predicts an essentially quiescent flow.

Finally, we test the models on sinusoidal inputs $w(t) = a\sin(k\omega t)$ applied at the mid-chord station~4, with amplitudes $a\in\{0.3,\,1\}$ and frequencies $k\omega$ with $k\in\{1,2\}$, where $\omega \approx 0.8026$ is the vortex-shedding frequency estimated from the training data.
These tests probe the ability of the models to extrapolate beyond the training data, since the models were trained exclusively on unforced impulse responses.
The results are shown in Figure~\ref{fig:airfoil_forced}.
At frequency $\omega$ and amplitude $a = 0.3$, the BGasOpInf model reproduces both the early-time growth and the saturated energy of the forced response almost exactly, while the GasOpInf model underestimates the energy by approximately four orders of magnitude.
At the larger amplitude $a = 1$, the BGasOpInf model slightly overshoots the energy of the FOM at early times and settles onto a saturated state whose energy is slightly lower than that of the FOM, while the GasOpInf model remains approximately four orders of magnitude off.
At frequency $2\omega$, the response of the FOM exhibits a transient hump before settling onto a lower-amplitude periodic state.
The BGasOpInf model accurately captures the transient hump at both amplitudes, which the GasOpInf model misses completely, and it predicts the energy of the post-transient periodic state almost exactly for $a = 0.3$, and within a factor of approximately three for $a = 1$.
The vorticity snapshots in Figure~\ref{fig:airfoil_snapshots}(b) confirm that, at frequency $\omega$, the BGasOpInf model reproduces the vortical structures shed into the wake, while the GasOpInf model severely underestimates their amplitude.
Overall, the BGasOpInf model is accurate across the vast majority of the tests considered here, including impulses at unseen forcing locations and sinusoidal inputs that lie outside the training regime, and, even in the most challenging cases, it outperforms the POD-based GasOpInf model by orders of magnitude.

\begin{figure}
  \centering
  \includegraphics[width=0.85\textwidth]{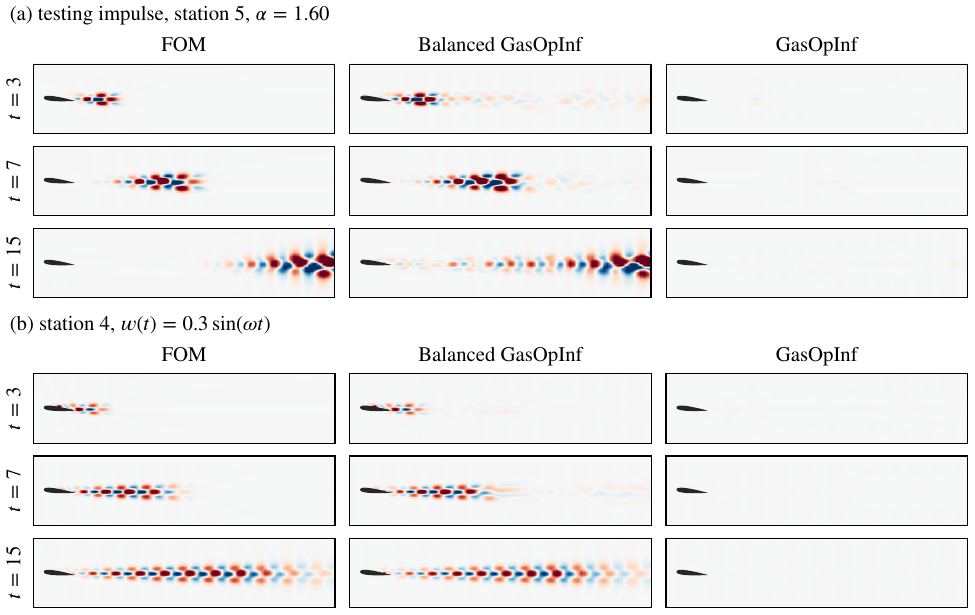}
  \caption{Airfoil flow. Vorticity snapshots of the FOM and of the reduced-order models (a) along the testing impulse response at the unseen station~5 with $\alpha = 1.60$, and (b) in response to the sinusoidal input $w(t) = 0.3\sin(\omega t)$ at station~4.}
  \label{fig:airfoil_snapshots}
\end{figure}

\FloatBarrier
\subsection{Flow Over a Finite Wing}
\label{subsec:wing}

Our last example is the three-dimensional incompressible flow over a finite wing with NACA 0012 cross section, unit chord and unit span (i.e., aspect ratio one) with closed end caps, at an angle of attack of $6^\circ$ and at a chord-based Reynolds number $Re = 1000$.
We include this example to demonstrate that the proposed method can be applied to large-scale three-dimensional problems, where the state dimension approaches $100$ million degrees of freedom.
The flow is simulated with the same immersed boundary projection solver used for the airfoil in Section~\ref{subsec:airfoil}, on the (non-dimensional) computational domain $[-3,15]\times[-4,4]\times[-4,4]$ with the half-chord of the wing located at the origin.
A uniform velocity $(1,0,0)$ is imposed at the inflow and at the four lateral boundaries, and a homogeneous Neumann condition is imposed at the outflow.
The mesh contains $400\times 260\times 260$ cells, with a uniform spacing of $0.0125$ in the region $[-0.75,0.75]^3$ that contains the wing, and it is stretched to a spacing of $0.1$ near the boundaries of the domain.
The state therefore has dimension $n = 3\,n_x n_y n_z\approx 8.08\times 10^7$, which is more than two orders of magnitude larger than in the airfoil example.
At these parameters, the flow admits a stable steady state (in accordance with the results in \cite{Pandi_Mittal_2023}), characterized by a pair of counter-rotating tip vortices, which we compute via time marching and show in Figure~\ref{fig:wing_base_modes}(a).
As in Section~\ref{subsec:airfoil}, we consider the dynamics of perturbations about the steady state, and we work with the weighted state $\widetilde{\vq} = \mW^{1/2}\vq$, where $\mW$ contains the ratio between the volume of each cell and the volume of the smallest cell. 
As in the previous example, we measure the whole weighted state, i.e., $\vy = \widetilde{\vq}$.

The perturbations are forced by inputs that act on the streamwise momentum equation with the spatial profile
\begin{equation}
  \label{eq:wing_forcing}
  b(x,y,z) = \exp\left\{-1250\left((x - x_c)^2 + (y - y_c)^2\right)\right\}h(z),
\end{equation}
where $h$ is a $C^2$ spanwise top-hat function that is equal to one for $\lvert z\rvert\leq 0.06$ and vanishes for $\lvert z\rvert\geq 0.10$.
We consider three stations located just above the suction side of the wing at the mid-span near the leading edge, the mid-chord and the trailing edge (stations 1, 2 and 3 in Figure~\ref{fig:wing_base_modes}(a)).
As for the airfoil, each forcing profile is projected onto the space of discretely divergence-free fields and normalized to unit norm, and it defines a column of the input matrix $\mB\in\R^{n\times 3}$.

\begin{figure}
  \centering
  \includegraphics[width=0.85\textwidth]{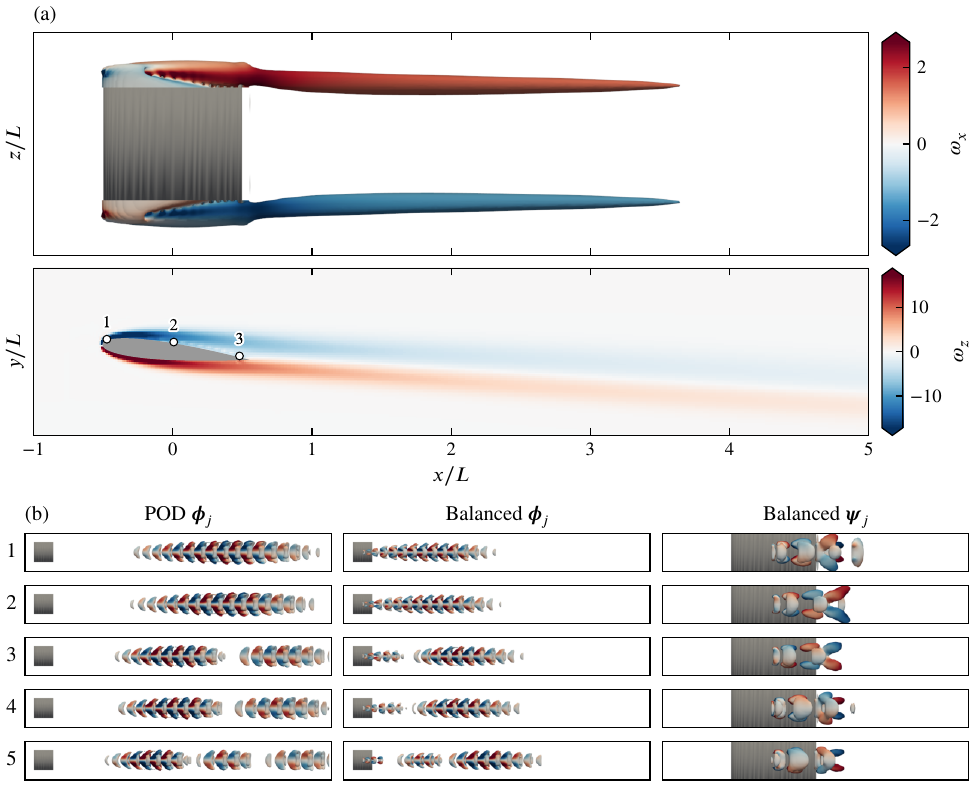}
  \caption{Finite wing. (a) Steady state at $Re = 1000$: isosurface $Q = 0.5$ of the $Q$-criterion colored by the streamwise vorticity $\omega_x$, viewed from above the wing (top), and spanwise vorticity $\omega_z$ in the mid-span plane, with the three forcing stations (bottom). (b) Leading POD modes (left), and columns $\vphi_j$ and $\vpsi_j$ of the trial and test matrices $\mPhi$ and $\mPsi$ computed by Algorithm~\ref{alg:checkpointed_balancing} (middle and right), viewed from above the wing. The test modes are shown in a close-up of the region around the wing.}
  \label{fig:wing_base_modes}
\end{figure}

\paragraph{Training.}
We train on the impulse responses from the initial conditions $\vq^{(j,k)}(0) = \alpha_k\mB_j$, with $j\in\{1,2,3\}$ and $\alpha_k\in\{0.1,\,1,\,2\}$, for a total of $s\,m = 9$ training trajectories over $t\in[0, 39.8]$.
As in the airfoil example, both the balancing procedure and POD are applied to the weighted snapshots $\widetilde{\vq}(t_i) = \mW^{1/2}\vq(t_i)$.
Snapshots are stored every $0.04$ time units for $t\in[0,4)$, and every $0.2$ time units for $t\in[4,39.8]$, for a total of $280$ snapshots per trajectory, and the empirical reachability Gramian is factored using trapezoidal quadrature weights as in the airfoil example.
The checkpointing procedure follows the one described in Section~\ref{subsec:airfoil}: we place $50$ checkpoints every $0.04$ time units on $[0,2)$ to resolve the fast transient that follows the impulse, and $18$ additional checkpoints approximately uniformly spaced on $[2,18]$, for a total of $N_c = 68$ checkpoints per trajectory.
All subtrajectories are defined over a window of length $\Delta T = 21.8$, the integrals in \eqref{eq:G_checkpoint} are evaluated with the rectangle rule on a uniform grid with spacing $0.2$, and subtrajectory samples that do not coincide with a stored snapshot are computed via six-point local Lagrange interpolation.
Since each training trajectory occupies approximately $65$\,GB of memory, and the GPU RAM was limited to $100$\,GB, we first performed a POD-based lossless compression of the training data, so that all the training was then performed in the compressed coordinates.
The wall-clock times for this example are reported in Table~\ref{tab:timings_wing}.
Since POD and balancing are carried out in the compressed coordinates, we report separately the cost of the lossless compression of the training data (i.e., a single pass over the stored snapshots to form their Gram matrix and its eigendecomposition), and the costs of POD and balancing in the compressed coordinates.
As in the airfoil example, the cost of computing the balanced bases is comparable to that of computing POD modes, and it is small compared to the cost of generating the training data.

\begin{table}
  \centering
  \caption{Finite wing. Wall-clock times (in seconds). The FOM was simulated on an NVIDIA GH200 GPU, and the ROM time refers to the integration of one trajectory over the training horizon $t\in[0,40]$. The lossless compression requires a single pass over the stored snapshots, and it is shared by POD and balancing, which are then carried out in the compressed coordinates.}
  \label{tab:timings_wing}
  \begin{tabular}{cccccc}
    \toprule
    Lossless compression & POD & Balancing (Alg.~\ref{alg:checkpointed_balancing}) & OpInf fit & ROM integration & FOM integration \\
    \midrule
    $215$ & $0.35$ & $1.09$ & $0.72$ & $0.82$ & $1459$ \\
    \bottomrule
  \end{tabular}
\end{table}
We take $r = 50$, for which the leading balanced modes capture more than $99\%$ of the energy of the Hankel singular values and the leading POD modes capture more than $99\%$ of the variance of the training data.
In this example, the latent-space dynamics are inferred with the standard quadratic Operator Inference formulation \eqref{eq:operator_inference} with reprojection, and with Tikhonov regularization (with regularization parameter $\lambda = 10^{-6}$) on the third-order tensor governing the quadratic interactions.
Regularization is included in the objective of \eqref{eq:operator_inference} as an additional term $\lambda\lVert\mH_r\rVert_F^2$, and it has been shown to encourage stability of the inferred dynamics.
We found that the stability-preserving formulation used in the airfoil example was not necessary here, although some regularization was required to prevent the unregularized models from blowing up.
Furthermore, since the input matrix is known, the input term is computed by projection, i.e., $\mB_r = \mPsi^\top\mW^{1/2}\mB$, rather than inferred.
As in the CGL examples, we refer to the resulting model as BOpInf, and we compare it with an OpInf model of the same dimension trained on the same data and with the same regularization, but where the high-dimensional data are projected orthogonally onto the span of the leading POD modes.
Since the snapshots are sampled non-uniformly in time, the POD modes are computed from the same non-uniformly sampled training data used for balancing, with each weighted snapshot $\widetilde{\vq}(t_i)$ scaled by $\sqrt{\gamma_i}$, where $\gamma_i$ are the same trapezoidal quadrature weights used for the empirical reachability Gramian, so that the POD modes are representative of the time-integrated (rather than sample-averaged) energy of the data.
The leading POD modes and the columns of $\mPhi$ and $\mPsi$ are shown in Figure~\ref{fig:wing_base_modes}(b).
The picture is consistent with the one observed for the airfoil: the POD modes are wavepackets supported far downstream in the wake, where the training trajectories carry most of their energy, the trial modes are wavepackets that are supported closer to the wing, and the test modes are localized on the suction side of the wing and in its immediate wake, where the flow is most receptive to perturbations.

\paragraph{Testing.}
We test the models on impulse responses forced at the three training stations with amplitudes that were not seen during training, namely $\alpha = 0.77$ at station~1, $\alpha = 1.5$ at station~2 and $\alpha = 1.82$ at station~3, and on a sinusoidal input $\mB\vu(t) = \left(0.5\mB_2 + 1.25\mB_3\right)\sin(2\pi f t)$ applied from rest, with $f = 0.8$.
This frequency corresponds to the peak of the energy-weighted spectrum of the training trajectories, and it is consistent with the frequency of the wavepackets observed in the airfoil example.
We recall that the models were trained exclusively on unforced impulse responses, so the sinusoidal input tests the ability of the models to extrapolate beyond the training data.
Additionally, the training trajectories were generated by individually perturbing each of the three training stations, while the sinusoidal input excites two stations simultaneously, so this test also probes the ability of the models to capture the nonlinear interactions between different inputs.
The time history of the perturbation energy is shown in Figure~\ref{fig:wing_testing}, and representative vorticity snapshots are shown in Figure~\ref{fig:wing_snapshots}.
For all three impulses, the BOpInf model follows the energy of the FOM almost exactly throughout the transient, including the initial growth, the plateau and the decay of the response as the wavepacket is advected out of the domain, while the OpInf model captures the transient reasonably well only at station~1, and it underestimates the energy of the response by a factor of approximately two at station~3 and by approximately one order of magnitude at station~2.
Once the perturbation of the FOM has decayed ($t\gtrsim 20$), the BOpInf model retains a small residual energy which, as in the airfoil example, lies two to three orders of magnitude below the peak of the response, while the OpInf model exhibits a spurious late-time growth.
These observations are confirmed by the snapshots in Figure~\ref{fig:wing_snapshots}(a,b): the BOpInf model predicts the correct location, shape and amplitude of the wavepacket as it is advected downstream, while the OpInf model either predicts a wavepacket that travels too fast (at station~1), or an essentially quiescent flow (at station~2).
In response to the sinusoidal input, the BOpInf model reaches the correct periodic state at the correct time, while the OpInf model underestimates the energy by a factor of approximately three and drifts instead of settling onto a periodic state.
The corresponding snapshots in Figure~\ref{fig:wing_snapshots}(c) show that BOpInf correctly reproduces the train of vortical structures shed into the wake, while the structures predicted by OpInf are significantly weaker and detached from the wing.

\begin{figure}
  \centering
  \includegraphics[width=0.85\textwidth]{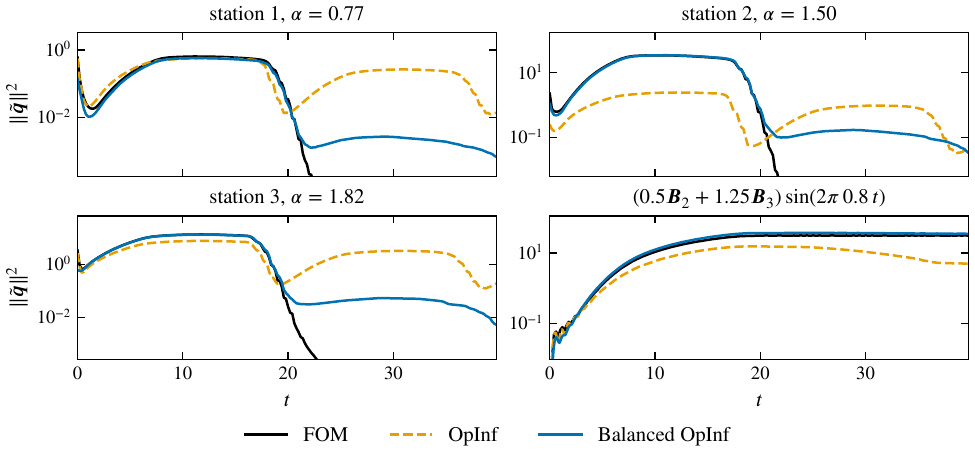}
  \caption{Finite wing. Perturbation energy of the testing impulse responses at the three training stations, and of the response to the sinusoidal input $\left(0.5\mB_2 + 1.25\mB_3\right)\sin(2\pi f t)$ with $f = 0.8$ (bottom right).}
  \label{fig:wing_testing}
\end{figure}

\begin{figure}
  \centering
  \includegraphics[width=0.85\textwidth]{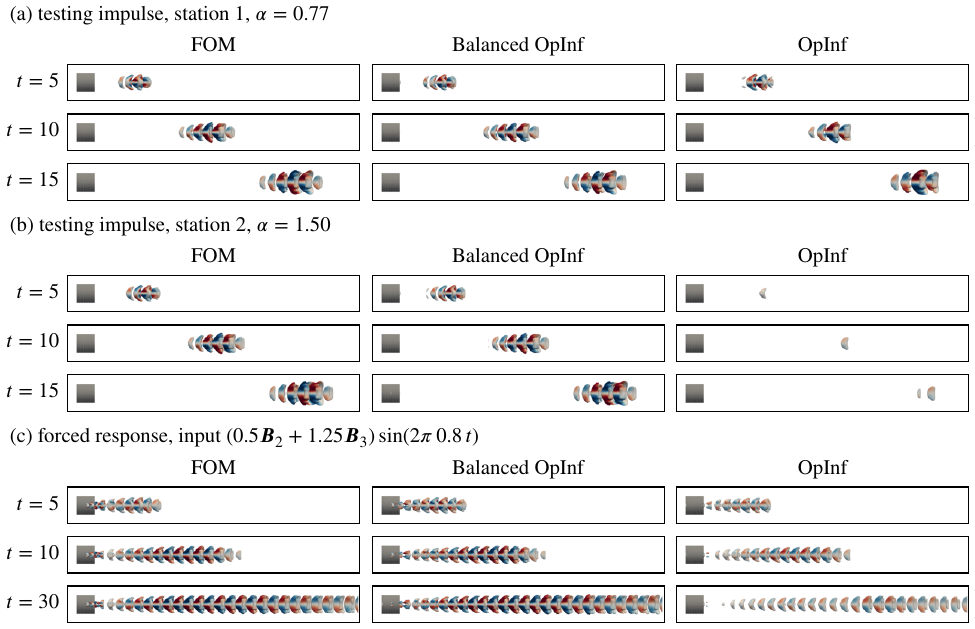}
  \caption{Finite wing. Snapshots of the FOM and of the reduced-order models (top view) along (a) the testing impulse response at station~1 with $\alpha = 0.77$, (b) the testing impulse response at station~2 with $\alpha = 1.50$, and (c) the response to the sinusoidal input in Figure~\ref{fig:wing_testing}.}
  \label{fig:wing_snapshots}
\end{figure}

\FloatBarrier
\section{Conclusions}
\label{sec:conclusions}

In this paper, we have introduced a non-intrusive and \emph{optimization-free} framework to compute balancing-based oblique projections for model reduction of high-dimensional nonlinear systems that exhibit large-amplitude transient energy growth.
State-of-the-art non-intrusive methods typically achieve dimensionality reduction via orthogonal projection onto high-energy linear or nonlinear manifolds, and this often leads to the truncation of low-energy, yet dynamically important, states.
While this issue is classically addressed by Balanced Truncation and its extensions, these methods are intrusive, since they leverage adjoint information to identify dynamically-important subspaces.
Here, we have addressed this issue by computing balancing-based oblique projections from high-dimensional data, leveraging a modified version of the empirical Gramians introduced by \cite{Lall1999} for nonlinear systems.
Specifically, while the computation of the empirical observability Gramian, as defined in \cite{Lall1999}, is intractable in high dimensions, its projection onto a suitably chosen restriction subspace can be computed inexpensively from forward simulation data.
We have shown that this procedure provably recovers the oblique projection defined via Balanced Truncation when the dynamics are linear and when the range of either of the Gramians is contained in the restriction subspace.
Furthermore, in the linear case and in the absence of adjoint information, we have shown that under a notion of minimax optimality the restriction subspace should span the dominant eigenspaces of the reachability Gramian, and this result informed the development of an efficient checkpointing algorithm that identifies the restriction subspace within the range of the already available impulse responses.
The overall cost of the procedure is therefore comparable to that of a standard SVD of the training data.
Once the balanced encoder-decoder pairs are computed, we learn the latent-space dynamics via Operator Inference or its globally-asymptotically-stable variant, although we remark once more that any other method to learn latent-space dynamics can be used.
The proposed framework was tested on several examples of increasing size and complexity, namely the complex Ginzburg-Landau equation in both stable and unstable configurations, the two-dimensional flow past a NACA 0012 airfoil at $Re = 600$, and the three-dimensional flow past a finite wing at $Re = 1000$ with approximately $8\times 10^7$ degrees of freedom.
In all these examples, the Balanced Operator Inference models accurately captured the large-amplitude transient response of the full-order model to impulses and to sinusoidal inputs that were never seen during training, and they significantly outperformed their POD-based counterparts, with a reduction in forecasting error of up to two orders of magnitude.

\section*{Acknowledgments}
The author gratefully acknowledges startup support from the Newark College of Engineering at New Jersey Institute of Technology.
All computations were performed on NCSA DeltaAI through NSF ACCESS allocation PHY260088.
The author also acknowledges the use of Claude Code powered by Opus 5.5 to assist with the numerical implementation of the algorithms and the data post-processing.
The author takes full responsibility for the results presented in this manuscript.

\printcredits

\appendix

\section{Proof of Theorem \ref{thm:balanced_truncation_bounds}}
\label{sec:proof_of_thm_1}

The first half of the proof follows directly from Weyl's inequality $\lvert\sigma_k\left(\mH + \Delta\right) - \sigma_k\left(\mH\right)\rvert \leq \sigma_1(\Delta)$ with $\widetilde{\mH} = \mH + \Delta$ and $\Delta \coloneqq \widetilde{\mH} - \mH = \mY^\top\mPi\mX - \mY^\top\mX = -\mY^\top\mPi_\perp\mX$.
  Then, using the submultiplicative property of the spectral norm and the fact that $\mPi_\perp$ is idempotent, we have, as desired,
  \begin{equation}
    \begin{aligned}
      \sigma_1(\Delta) = \lVert \mY^\top\mPi_\perp\mX\rVert &\leq \lVert \mY^\top\mPi_\perp\rVert\lVert \mPi_\perp\mX\rVert = \lVert \mPi_\perp \underbrace{\mY\mY^\top}_{\mW_O} \mPi_\perp\rVert^{1/2} \lVert \mPi_\perp \underbrace{\mX\mX^\top}_{\mW_R} \mPi_\perp\rVert^{1/2} = \delta_O \delta_R.
    \end{aligned}
  \end{equation}
  The second half of the proof relies on the residual form of Wedin's $\sin\theta$ theorem \citep[Theorem~V.4.4]{stewart1990matrix}. 
  Applying Wedin's theorem to the leading $r$ singular triplets of $\mH$ and $\widetilde{\mH}$, we have (directly from the theorem statement)
  \begin{equation}
    \label{eq:wedin}
    \max\big\{\sin\theta_{\max}\left(\mU_r,\widetilde{\mU}_r\right), \sin\theta_{\max}\left(\mV_r,\widetilde{\mV}_r\right)\big\} \leq \varepsilon^{-1}\max\{\lVert \mathcal{R}\rVert, \lVert \mathcal{L}\rVert\},
  \end{equation}
  provided $\varepsilon>0$, where $\mathcal{R} \coloneqq \left(\mH - \widetilde{\mH}\right)\widetilde{\mV}_r$ and $\mathcal{L} \coloneqq \left(\mH - \widetilde{\mH}\right)^\top\widetilde{\mU}_r$ are the residuals of the leading $r$ singular triplets of $\widetilde{\mH}$ with respect to $\mH$.
  Using the orthonormality of the columns of $\widetilde{\mV}_r$ and the bound in the first half of the proof, we have $\lVert \mathcal{R}\rVert = \big\lVert \left(\mH - \widetilde{\mH}\right)\widetilde{\mV}_r\big\rVert \leq \lVert \mH - \widetilde{\mH}\rVert = \lVert\Delta\rVert \leq \delta_O\delta_R$, and similarly $\lVert \mathcal{L}\rVert \leq \delta_O\delta_R$.
  Thus, the right-hand side of~\eqref{eq:wedin} can be bounded by $\varepsilon^{-1}\delta_O\delta_R$, and, in particular, $\sin\theta_{\max}(\mV_r,\widetilde{\mV}_r)\leq\varepsilon^{-1}\delta_O\delta_R$.

  To prove \eqref{eq:sin_theta_bound}, let $\mW_R=\mX\mX^\top$ be positive definite and define $\mT\coloneqq\mW_R^{-1/2}\mX$, so that $\mW_R^{-1/2}\mPhi = \mT\mV_r\mSigma_r^{-1/2}$ and $\mW_R^{-1/2}\widetilde{\mPhi} = \mT\widetilde{\mV}_r\widetilde{\mSigma}_r^{-1/2}$.
  Since the diagonal matrices $\mSigma_r^{-1/2}$ and $\widetilde{\mSigma}_r^{-1/2}$ do not affect the ranges, by definition we have $\theta^{\mW_R^{-1}}_{\max}(\mPhi,\widetilde{\mPhi}) = \theta_{\max}(\mT\mV_r,\mT\widetilde{\mV}_r)$.
  Now, $\mT^\top\mT = \mX^\top\left(\mX\mX^\top\right)^{-1}\mX = \mathbb{P}_{\mX^\top}$ is the orthogonal projector onto $\range(\mX^\top)$, so $\mT$ preserves inner products (and hence norms) between vectors in $\range(\mX^\top)$.
  Since $\mV_r$ and $\widetilde{\mV}_r$ are right singular vectors of $\mY^\top\mX$ and $\mY^\top\mPi\mX$ associated with nonzero singular values, then $\range(\mV_r),\range(\widetilde{\mV}_r)\subseteq\range(\mX^\top)$.
  It follows that $\mT\mV_r$ and $\mT\widetilde{\mV}_r$ have orthonormal columns, and that $\mZ\coloneqq\left(\mI - \mV_r\mV_r^\top\right)\widetilde{\mV}_r$ has columns in $\range(\mX^\top)$, so that $\lVert\mT\mZ\rVert = \lVert\mZ\rVert$.
  Therefore,
  \begin{equation}
    \sin\theta_{\max}\left(\mT\mV_r,\mT\widetilde{\mV}_r\right) = \big\lVert\left(\mI - \mT\mV_r\mV_r^\top\mT^\top\right)\mT\widetilde{\mV}_r\big\rVert = \big\lVert\mT\mZ\big\rVert = \big\lVert\mZ\big\rVert = \sin\theta_{\max}\left(\mV_r,\widetilde{\mV}_r\right)\leq \varepsilon^{-1}\delta_O\delta_R,
  \end{equation}
  where the second equality follows from $\mT^\top\mT\widetilde{\mV}_r = \widetilde{\mV}_r$, and this proves the first half of \eqref{eq:sin_theta_bound}.
  To prove the second half of \eqref{eq:sin_theta_bound}, we observe that $\mX^\top\mY\mU_r = \mH^\top\mU_r = \mV_r\mSigma_r$ and $\mX^\top\mPi\mY\widetilde{\mU}_r = \widetilde{\mH}^\top\widetilde{\mU}_r = \widetilde{\mV}_r\widetilde{\mSigma}_r$, so that
  \begin{equation}
    \mW_R\mPsi = \mX\mX^\top\mY\mU_r\mSigma_r^{-1/2} = \mX\mV_r\mSigma_r^{1/2},\qquad \mW_R\widetilde{\mPsi} = \mX\mX^\top\mPi\mY\widetilde{\mU}_r\widetilde{\mSigma}_r^{-1/2} = \mX\widetilde{\mV}_r\widetilde{\mSigma}_r^{1/2}.
  \end{equation}
  Multiplying on the left by $\mW_R^{-1/2}$ gives $\mW_R^{1/2}\mPsi = \mT\mV_r\mSigma_r^{1/2}$ and $\mW_R^{1/2}\widetilde{\mPsi} = \mT\widetilde{\mV}_r\widetilde{\mSigma}_r^{1/2}$, so that $\theta^{\mW_R}_{\max}(\mPsi,\widetilde{\mPsi}) = \theta_{\max}(\mT\mV_r,\mT\widetilde{\mV}_r) = \theta_{\max}(\mV_r,\widetilde{\mV}_r)$ by the same argument used above, and this concludes the proof.

\bibliographystyle{cas-model2-names}

\bibliography{references}



\end{document}